\documentclass[a4paper,UKenglish]{lipics-v2016}
 
\usepackage{microtype}
\usepackage{microtype}
\usepackage{amsmath}
\usepackage{amsfonts}

\title{On the Parameterized Complexity of Contraction to Generalization of Trees}
\titlerunning{Contraction to Generalization of Trees} 

\author[1]{Akanksha Agrawal}
\author[1,2]{Saket Saurabh}
\author[2]{Prafullkumar Tale}
\affil[1]{Department of Informatics, University of Bergen, Bergen, Norway \\
  \texttt{akanksha.agrawal@uib.no}}
\affil[2]{The Institute of Mathematical Sciences, HBNI, Chennai, India\\
  \texttt{{saket|pptale}@imsc.res.in}}

\authorrunning{A. Agrawal, S. Saurabh and P. Tale} 

\Copyright{A. Agrawal, S. Saurabh and P. Tale} 


\EventEditors{John Q. Open and Joan R. Acces}
\EventNoEds{2}
\EventLongTitle{The 12th International Symposium on Parameterized and Exact Computation (IPEC 2017)}
\EventShortTitle{IPEC 2017}
\EventAcronym{IPEC}
\EventYear{2017}
\EventDate{}
\EventLocation{}
\EventLogo{}
\SeriesVolume{}
\ArticleNo{}

\newcommand{\calW}{\mathcal{W}}
\newcommand{\bbT}{\mathbb{T}_{\ell}}
\newcommand{\treelshort}{\mathbb{T}_{\ell}\textsc{C}}

\newtheorem{observation}{Observation}
\newtheorem{proposition}{Proposition}
\newtheorem{reduction rule}{Reduction Rule}[section]
\newtheorem{branching rule}{Branching Rule}[section]

\newcommand{\FPT}{\textsf{FPT}}
\newcommand{\NP}{\textsf{NP}}
\newcommand{\NPH}{\textsf{NP}-hard}
\newcommand{\NPC}{\textsf{NP}-complete}

\newcommand{\CONPpoly}{\textsf{coNP/poly}}

\newcommand{\no}{\emph{no}}
\newcommand{\yes}{\emph{yes}}

\newcommand{\ncol}{2\sqrt{\ell}+2}
\usepackage{todonotes}

\newcommand{\cvcfull}{\textsc{Connected Vertex Cover}}
\newcommand{\cvc}{\textsc{CVC}}
\newcommand{\treec}{$\mathbb{T}$-\textsc{Contraction}}
\newcommand{\treeloc}{$\mathbb{T}_{\ell_1}$-\textsc{Contraction}}
\newcommand{\treeltc}{$\mathbb{T}_{\ell_2}$-\textsc{Contraction}}

\newcommand{\calF}{\mathcal{F}}

\newcommand{\OO}{\mathcal{O}}
\newcommand{\treelc}{$\mathbb{T}_\ell$-\textsc{Contraction}}
\newcommand{\treelpc}{$\mathbb{T}_{\ell'}$-\textsc{Contraction}}

\newcommand{\defparproblem}[4]{
 \vspace{3mm}
\noindent\fbox{
  \begin{minipage}{.95\textwidth}
  \begin{tabular*}{\textwidth}{@{\extracolsep{\fill}}lr} \textsc{#1}  & {\bf{Parameter:}} #3 \\ \end{tabular*}
  {\bf{Input:}} #2  \\
  {\bf{Question:}} #4
  \end{minipage}
  }
  \vspace{2mm}
}

\begin{document}

\maketitle

\begin{abstract}
For a family of graphs $\cal F$, the $\mathcal{F}$-\textsc{Contraction} problem takes as an input a graph $G$ and an integer $k$, and the goal is to decide if there exists $S \subseteq E(G)$ of size at most $k$ such that $G/S$ belongs to $\cal F$. Here, $G/S$ is the graph obtained from $G$ by contracting all the edges in $S$. Heggernes et al.~[{\em Algorithmica (2014)}] were the first to study edge contraction problems in the realm of Parameterized Complexity. They studied $\cal F$-\textsc{Contraction}  when $\cal F$ is a simple family of graphs such as trees and paths. In this paper, we study the $\mathcal{F}$-\textsc{Contraction} problem, where $\cal F$ generalizes the family of trees. In particular, we define this generalization in a ``parameterized way''. Let $\mathbb{T}_\ell$ be the family of graphs such that each graph in $\mathbb{T}_\ell$ can be made into a tree by deleting at most $\ell$ edges. Thus, the problem we study is \treelc. We design an \FPT\ algorithm for  \treelc\ running in time $\mathcal{O}((\ncol)^{\mathcal{O}(k + \ell)} \cdot n^{\mathcal{O}(1)})$. Furthermore, we show that the problem does not admit a polynomial kernel when parameterized by $k$. Inspired by the negative result for the kernelization, we design a lossy kernel for \treelc\ of size $ \mathcal{O}([k(k + 2\ell)] ^{(\lceil {\frac{\alpha}{\alpha-1}\rceil + 1)}})$.
\end{abstract}

\section{Introduction}
Graph editing problems are one of the central problems in graph theory that have been extensively studied in the realm of Parameterized Complexity. Some of the important graph editing operations are vertex deletion, edge deletion, edge addition, and edge contraction. For a family of graphs $\cal F$, the $\cal F$-\textsc{Editing} problem takes as an input a graph $G$ and an integer $k$, and the goal is to decide whether or not we can obtain a graph in $\cal F$ by applying at most $k$ edit operations on $G$. In fact, the $\cal F$-\textsc{Editing} problem, where the edit operations are restricted to one of vertex deletion, edge deletion, edge addition, or edge contraction have also received a lot of attention in Parameterized Complexity. When we restrict the operations to only deletion operation (vertex/edge deletion) then the corresponding problem is called $\cal F$-\textsc{Vertex (Edge) Deletion} problem.  On the other hand if we only allow edge contraction then the corresponding problem is called $\mathcal{F}$-\textsc{Contraction}. The \textsc{$\mathcal{F}$-Editing} problem generalizes several \NPH\ problems such as  \textsc{Vertex Cover}, \textsc{Feedback vertex set}, {\sc Planar $\cal F$-Deletion}, {\sc Interval Vertex Deletion}, {\sc Chordal Vertex Deletion}, \textsc{Odd cycle transversal}, {\sc Edge Bipartization}, {\sc Tree Contraction}, {\sc Path Contraction}, {\sc Split Contraction}, {\sc Clique Contraction}, etc. Most of the studies in the Parameterized Complexity or the classical Complexity Theory have been restricted to combination of vertex deletion, edge deletion or edge addition. Only recently, edge contraction as an edit operation has started to gain attention in the realm of Parameterized Complexity.  In this paper, we add another family of graphs $\cal F$ -- a parameterized generalization of trees --  such that $\mathcal{F}$-\textsc{Contraction} is fixed parameter tractable (\FPT). We also explore the problem  from the viewpoints of  Kernelization Complexity as well as its new avatar the Lossy Kernelization. For more details on Parameterized Complexity we refer to the books of Downey and Fellows~\cite{DF97,DF-new}, Flum and Grohe~\cite{flumgrohe}, Niedermeier~\cite{niedermeier2006}, and Cygan et al.~\cite{saurabh-book}. 

Our starting point is the result of Heggernes et al.~\cite{treePathContract} who studied $\mathcal{F}$-\textsc{Contraction} when  $\cal F$ is the family of paths ($\mathbb{P}$) and trees ($\mathbb{T}$). To the best of our knowledge these were the first results concerning Parameterized Complexity of $\mathcal{F}$-\textsc{Contraction}  problems. They showed that $\mathbb{P}$-\textsc{Contraction} and 
$\mathbb{T}$-\textsc{Contraction} are \FPT. Furthermore, they showed that $\mathbb{T}$-\textsc{Contraction} does not admit a polynomial kernel. On the other hand $\mathbb{P}$-\textsc{Contraction} admits a polynomial kernel with at most  $5k+3$ vertices  (see~\cite{LiFCH17} for an improved bound of $3k+4$ on the number of vertices). Moreover, $\mathcal{F}$-\textsc{Contraction} is not \FPT (unless some unlikely collapse in Parameterized Complexity happens) even for simple family of graphs such as $P_t$-free graphs for some $t \geq 5$,  the family of $C_t$-free graphs for some $t \geq 4$~\cite{CliqueContractUpp,elimNew}, and the family of split graphs~\cite{AgrawalLSZ17}. Here, $P_t$ and $C_t$ denotes the path and cycle on $t$ vertices. In light of these mixed answers, two natural questions are: 
\begin{enumerate}
\item What additional parameter we can associate with $\mathbb{T}$-\textsc{Contraction} such that it admits a polynomial kernel? 
\item What additional parameter we can associate with $\mathbb{T}$-\textsc{Contraction} such that an \FPT\ algorithm with combination of these parameterizations leads to an algorithm that generalizes the \FPT\ algorithm on trees? 
\end{enumerate}
In our earlier paper (a superset of authors) we addressed the first question~\cite{AgarwalLST17}. In particular we studied  $\calF$-\textsc{Contraction}, where $\cal F$ is the {\em family of trees with at most $\ell$ leaves} (together with some other problems), and designed a polynomial kernel (hence an \FPT\ algorithm) with $\OO(k \ell)$ vertices. This was complimented by a matching kernel lower bound result. In this paper we focus on the second question. 

\subparagraph*{Our Problem and Results.} To define our problem formally let us define $\mathbb{T}_\ell$ to be the family of graphs such that each graph in $\mathbb{T}_\ell$ can be made into a tree by deleting at most $\ell$ edges. Thus the problem we study will be called \treelc. 

\defparproblem{\treelc}{A graph $G$ and an integer $k$.}{$k$}{Does there exist $S \subseteq E(G)$ of size at most $k$ such that $G/S \in \mathbb{T}_\ell$?}

Observe that for $\ell=0$, \treelc\ is the usual \treec. 
We design an \FPT\ algorithm for \treelc\ running in time $\mathcal{O}((\ncol)^{\mathcal{O}(k + \ell)} \cdot n^{\mathcal{O}(1)})$. Our algorithm follows the general approach of Heggernes et al.~\cite{treePathContract} for designing the algorithm for \treec. Also, we show that the problem does not admit a polynomial kernel, when parameterized by $k$, for any (fixed) $\ell \in \mathbb{N}$. Inspired by the negative result on kernelization, we design a lossy kernel for \treelc.

\subparagraph*{Related Works.} For several families of graphs $\cal F$, early papers by Watanabe et al.~\cite{contractEarly2,contractEarly3} and Asano and Hirata \cite{asano1983edge} showed that {\sc ${\cal F}$-Contraction} is \NP-complete. From the viewpoint of Parameterized Complexity these problems exhibit properties that are quite different from the problems where  the edit operations are restricted to deleting or adding vertices or edges. For instance, deleting $k$ edges from a graph such that the resulting graph is a tree is polynomial time solvable. On the other hand, Asano and Hirata showed that \treec\  is \NPH~\cite{asano1983edge}. Furthermore, a well-known result by Cai~\cite{cai1996} states that when $\cal F$ is a hereditary family of graphs with a finite set of forbidden induced subgraphs  then the graph modification problem defined by $\cal F$ and the edit operations restricted to vertex deletion, edge deletion, or edge addition admits an \FPT\ algorithm. Moreover, this result does not hold when the edit operation is edge contraction. Lokshtanov et al.~\cite{elimNew} and Cai and Guo \cite{CliqueContractUpp} independently showed that if $\cal F$ is either the family of $P_\ell$-free graphs for some $\ell\geq 5$ or the family of $C_\ell$-free graphs for some $\ell\geq 4$ then {\sc ${\cal F}$-Contraction} is W[2]-hard. Golovach et al.~\cite{planarContract} proved that if $\cal F$ is the family of planar graphs then {\sc ${\cal F}$-Contraction} is \FPT.  Belmonte et al.~\cite{Belmonte:2014} proved that the problem is \FPT\ for $\cal F$ being the family of degree constrained graphs like bounded degree,  (constant) degenerate and (constant) regular graphs. Moreover, Cai and Guo \cite{CliqueContractUpp} showed that in case $\cal F$ is the family of cliques, {\sc ${\cal F}$-Contraction} is solvable in time $2^{\OO(k\log k)}\cdot n^{\OO(1)}$, while in case $\cal F$ is the family of chordal graphs, the problem is W[2]-hard. Heggernes et al.~\cite{bipartiteContract} developed an \FPT\ algorithm for the case where $\cal F$ is the family of bipartite graphs (see~\cite{bipartiteContract2} for a faster algorithm). 

\section{Preliminaries} \label{sec:prelims-treelc}
In this section, we state some basic definitions and introduce terminologies from graph theory and algorithms. We also establish some of the notations that will be used throughout. We denote the set of natural numbers by $\mathbb{N}$ (including $0$). For $k \in \mathbb{N}$, by $[k]$ we denote the set $\{1,2,\ldots, k\}$. Let $X,Y$ be two sets. For a function $\varphi: X \rightarrow Y$ and $y \in Y$, by $\varphi^{-1}(y)$ we denote the set $\{x \in X \mid \varphi(x)=y\}$.

\subparagraph*{Graphs.} We use standard terminologies from the book of Diestel~\cite{diestel-book} for the graph related terms which are not explicitly defined here. We consider simple graphs. For a graph $G$, by $V(G)$ and $E(G)$ we denote the vertex and edge sets of $G$, respectively. For a vertex $v \in V(G)$, we use $deg_G(v)$ to denote the degree of $v$ in $G$, \emph{i.e.} the number of edges in $G$ that are incident to $v$. For $v \in V(G)$, by $N_G(v)$ we denote the set $\{u \in V(G) \mid vu \in E(G)\}$. We drop the subscript $G$ from $deg_G(v)$ and $N_G(v)$ whenever the context is clear. For a vertex subset $S \subseteq V(G)$, by $G[S]$ we denote the graph with the vertex set $S$ and the edge set as $\{vu \in E(G) \mid v, u \in S\}$. By $G - S$ we denote the graph $G[V(G) \setminus S]$. We say $S, S' \subseteq V(G)$ are \emph{adjacent} if there is $v \in S$ and $v' \in S'$ such that $vv' \in E(G)$. Further, an edge $uv \in E(G)$ is \emph{between} $S$ and $S'$ if $u \in S$ and $v \in S'$.

A sequence of vertices $P=(v_1,v_2, \ldots,v_q)$ is a \emph{path} in a graph $G$ if for all $i \in [q-1]$, $v_iv_{i+1} \in E(G)$. Furthermore, we call such a path as a path between $v_1$ and $v_q$. A graph is \emph{connected} if there is a path between every pair of its vertices. Otherwise, we call it a \emph{disconnected} graph. A maximal connected subgraph is called a \emph{component} in a graph. A graph is called $k$-\emph{vertex connected} or $k$-connected if for all $S' \subseteq V(G)$ such that the number of components in $G - S$ is more than the number of components in $G$ we have $|S'| \geq k$. That is, deleting a vertex subset of size less than $k$ cannot result in a graph with strictly more components.

A sequence of vertices $C=(v_1,v_2, \ldots,v_q)$ is a \emph{cycle} in a graph $G$ if for all $i \in [q]$, $v_iv_{i+1} \in E(G)$ (index computed modulo $q$). A graph is called a \emph{forest} or an \emph{acyclic graph} if it does not contain any cycle. A \emph{tree} is a connected acyclic graph. A vertex in a tree with exactly one neighbor is called a \emph{leaf}. The vertices in a tree which are not leaves are \emph{internal} vertices. 

For $\ell \in \mathbb{N}$, by $\bbT$ we denote the family of graphs from which we can obtain a tree using at most $\ell$ edge deletions. 
Observe that for any graph $G \in \bbT$, we have $|E(G)| \leq |V(G)| - 1 + \ell$. Moreover, for any connected graph $G$, if $|E(G)| \leq |V(G)| - 1 + \ell$ then $G \in \bbT$.

A vertex subset $S \subseteq V(G)$ is said to \emph{cover} an edge $uv \in E(G)$ if $S \cap \{u,v\} \neq \emptyset$. A vertex subset $S \subseteq V(G)$ is called a \emph{vertex cover} in $G$ if it covers all the edges in $G$. A \emph{minimum vertex cover} is a set $S\subseteq V(G)$ such that $S$ is a \emph{vertex cover} and for all $S'\subseteq V(G)$ such that $S'$ is a \emph{vertex cover}, we have $|S| \leq |S'|$. A \emph{vertex cover} $S$ in $G$ is said to be a \emph{connected-vertex cover} if $G[S]$ is a \emph{connected} graph. 

For $uv \in E(G)$, by \emph{contracting} the edge $uv$ in $G$ we mean the graph obtained by the following operation. We add a vertex $uv^{\star}$ and make it adjacent to all the vertices in $(N(v) \cup N(u)) \setminus \{v,u\}$ and delete $u,v$ from the graph. We often call such an operation as \emph{contraction} of the edge $uv$. For $E' \subseteq E(G)$, by $G/E'$ we denote the graph obtained from $G$ by contracting edges in $E'$. 

A graph $G$ is \emph{isomorphic} to a graph $H$ if there exists a \emph{bijective} function $\phi : V(G) \rightarrow V(H)$ such that for all $u,v \in V(G)$, $uv \in E(G)$ if and only if $\phi(u)\varphi(v) \in E(H)$. A graph $G$ is \emph{contractible} to a graph $H$, if their exists $E' \subseteq E(G)$ such that $G/E'$ is \emph{isomorphic} to $H$. In other words, $G$ is \emph{contractible} to $H$ if there exists a \emph{surjective} function $\varphi: V(G) \rightarrow V(H)$ with $W(h)= \{v \in V(G) \mid \varphi(v) =h\}$, for $h \in V(H)$ and the following property holds.
\begin{itemize}
\item For all $h,h' \in V(H)$, $hh' \in E(H)$ if and only if $W(h), W(h')$ are \emph{adjacent} in $G$.
\item For all $h \in V(H)$, $G[W(h)]$ is \emph{connected}.
\end{itemize}

Let $\mathcal{W}=\{W(h) \mid h \in V(H)\}$. Observe that $\mathcal{W}$ defines a partition of vertices in $G$. We call $\mathcal{W}$ as an $H$-\emph{witness structure} of $G$. The sets in $\mathcal{W}$ are called \emph{witness sets}. If a \emph{witness set} contains more than one vertex then we will call it a \emph{big} witness-set, otherwise it is a \emph{small} witness set. A graph $G$ is said to be $k$-contractible to a graph $H$ if there exists $E' \subseteq E(G)$ such that $G/E'$ is isomorphic to $H$ and $|E'| \leq k$. We use the following observation while designing our algorithms.

\begin{observation}
Let $G$ be a graph which is $k$-contractible to a graph $H$ and $\calW$ be an $H$-witness structure of $G$. Then the following holds.
\begin{itemize}
\item $|V(G)| \leq |V(H)|+ k$. 
\item For all $W \in {\cal W}$ we have $|W| \leq k+1$.
\item $\calW$ has at most $k$ big witness sets.
\item The union of the big witness sets in $\calW$ contains at most $2k$ vertices.
\end{itemize}
\end{observation}

A $k$-coloring of a graph $G$ is a function $\phi: V(G) \rightarrow [k]$. A $k$-coloring $\phi$ of $G$ is a \emph{proper coloring} if for all $uv \in E(G)$ we have $\phi(u) \neq \phi(v)$. The \emph{chromatic number} of a graph is the minimum number of colors needed for its proper coloring. For a subset $S \subseteq V(G)$ and a $k$-coloring $\phi$ of $G$, $S$ is said to be \emph{monochromatic} with respect to $\phi$ if for all $s,s' \in S$, $\phi(s)=\phi(s')$. Observe that $\phi$ partitions $V(G)$ into (at most) $k$ pairwise disjoint sets.
A subset $S \subseteq V(G)$ is said to be \emph{monochromatic component} with respect to $\phi$ if $S$ is monochromatic and $G[S]$ is connected.

\begin{lemma}\label{lemma:treel-witness-prop1} Let $F$ be an (inclusion-wise) minimal set of edges in a graph $G$ such that $T=G/F \in \bbT$ and $|V(G/F)| \geq 3$, and $\mathcal{W}$ be a $T$-witness structure of $G$. Then, the following properties hold.
\begin{enumerate}
\item There exists a set $F'$ of at most $|F|$ edges in $G$ such that $G/F'$ is in $\bbT$ and the $G/F'$-witness structure $\mathcal{W}'$ of $G$ satisfies the property that for every leaf $t$ in $G/F'$, $W'(t) \in \mathcal{W}'$ is a singleton set.
\item If $G$ is $2$-connected and $t$ is cut vertex in $T$ then $|W(t)| > 1$. 
\end{enumerate}
\end{lemma}
\begin{proof} 
\emph{(Proof of Part 1.)} If for each leaf $t \in V(T)$ we have $|W(t)| = 1$ then $F' = F$ is a desired solution. Otherwise, consider a leaf $t$ in $T$ such that $|W(t)| > 1$. Let $t'$ be the unique neighbour of $t$ in $T$. Notice that $W(t)$ and $W(t')$ are adjacent in $G$, and $G[W(t) \cup W(t')]$ is connected. Fix a spanning tree $Q$ of $G[W(t)]$, and (arbitrarily) choose a vertex $u^*\in V(Q)$ that is adjacent to a vertex in $W(t')$, which exists since $tt' \in E(T)$. Furthermore, choose a leaf $v^* \in V(Q) \setminus \{u^*\}$, which exists as $|V(Q)| >1$. Let $W'(t')=(W(t') \cup W(t)) \setminus \{v^*\}$, $W'(t)=\{v^*\}$, and ${\cal W}'=({\cal W} \setminus \{W(t), W(t')\}) \cup \{W'(t), W'(t')\}$. Notice that ${\cal W}'$ is a $T$-witness structure of $G$, and the number of leaves corresponding to singleton witness sets is strictly more than that of $\cal W$. Hence, by repeating this argument for each (non-adjacent) leaves in $T$ and their corresponding witness set in $\cal W$, we can obtain the desired result. 

\emph{(Proof of Part 2.)} Let $t$ be a cut vertex in $T$ such that $W(t) = \{u\}$, where $u \in V(G)$. Notice that $T - \{t\}$ has at least two components, say $T_1$ and $T_2$. Consider $U_1 = \bigcup_{t \in V(T_1)} W(t)$ and $U_2 = \bigcup_{t \in V(T_2)} W(t)$. As $\mathcal{W}$ is a $T$-witness structure of $G$, it follows that there is no edge between a vertex in $U_1$ and a vertex in $U_2$ in $G$. This contradicts the fact that $G$ is $2$-connected.
\end{proof}

\subparagraph*{Lossy Kernelization.} In lossy kernelization, we work with optimzation analogue of parameterized problem. Along with an instance and a parameter, optimization analogue of the problem also has a string called \emph{solution}. For a parameterized problem $Q$, its optimzation analogue is a computable function $\Pi: \Sigma^* \times \mathbb{N} \times \Sigma^* \rightarrow \mathbb{R} \cup \{\pm \infty\}$.
The {\em value} of a solution $S$ to an instance $(I,k)$ of $Q$ is $\Pi(I,k,S)$. In this paper, we will be dealing only with parameterized minimization problems. The {\em optimum value} for an instance $(I,k)$ of a \emph{minimization problem} $Q$ is denoted by $\textsc{OPT}_{\Pi}(I, k)$ and defined as $\min_{S \in \Sigma^*}\{ \Pi(I,k,S)\}$. An {\em optimum solution} for $(I,k)$ is a solution $S$ for which this optimum value is achieve. We omit the subscript $\Pi$ in the notation for optimum value if the problem under consideration is clear from the context.

\begin{definition}[$\alpha$-PTAS] An {\em $\alpha$-approximate polynomial-time preprocessing algorithm} ($\alpha$-PTAS) for a parameterized minimization problem $\Pi$ is pair of two polynomial time algorithms as follows:
\begin{enumerate}
\item \emph{Reduction Algorithm :} Given an instance $(I,k)$ of $\Pi$ it outputs an instance $(I',k')$ of $\Pi$. 
\item \emph{Solution Lifting Algorithm :} Given instances $(I,k)$ and $(I',k')$ of $\Pi$, and a solution $S'$ to $(I',k')$, it computes a solution $S$ to $(I,k)$ such that the following holds. $$\frac{\Pi(I,k,S)}{\textsc{OPT}(I,k)} \leq \alpha \cdot \frac{\Pi(I',k',S')}{\textsc{OPT}(I',k')}$$
\end{enumerate}
\end{definition}
\begin{definition}[$\alpha$-approximate kernelization] An {\em $\alpha$-approximate kernelization (or $\alpha$-approximate kernel)} for a parameterized minimization problem $\Pi$ is an $\alpha$-approximate polynomial-time preprocessing algorithm such that the size of the output instance is upper bounded by a computable function $g:\mathbb{N} \rightarrow \mathbb{N}$ of $k$.
\end{definition}

An $\alpha$-approximate kernelization for a parameterized minimization problem is said to be {\em strict} if the solution lifting algorithm returns a solution $S$ such that the following condition is satisfied.
$$\frac{\Pi(I,k,S)}{\textsc{OPT}(I,k)} \leq \max \{ \frac{\Pi(I',k',S')}{\textsc{OPT}(I',k')} , \alpha \}$$

A reduction rule is said to be {\em $\alpha$-safe} for $\Pi$ if there is a solution lifting algorithm such that the rule together with this algorithm constitute a strict $\alpha$-approximate polynomial-time preprocessing algorithm for $\Pi$. 
A {\em polynomial-size approximate kernelization scheme (PSAKS)} for $\Pi$ is a family of $\alpha$-approximate polynomial kernelization algorithms for each $\alpha >1$.
Since we are interested in solutions of size at most $k$, we work with following definition.
$$\Pi(I, k, S) = \left\{ \begin{array}{rl} \infty &\mbox{\text{ if } S \text{ is not a solution }} \\ \min\{|S|, k + 1\} &\mbox{ otherwise} \end{array} \right.$$ 
With this definition, if the solution lifting algorithm is given a solution of value $k+1$ or more, it simply returns any trivial feasible solution to the instance. We encourage the reader to see~\cite{lossy} for a more comprehensive discussion of these ideas and definitions.

\section{\FPT\ Algorithm for \treelc} \label{sec:rand-fpt-algo-treelc}
In this section, we design an \FPT\ algorithm for \treelc. Our algorithm proceeds as follows. We start by applying some simple reduction rules. Then by branching we ensure that the resulting graph is $2$-connected. Finally, we give an \FPT\ algorithm running in time $\mathcal{O}((\ncol)^{\mathcal{O}(k + \ell)} \cdot n^{\mathcal{O}(1)})$ on $2$-connected graphs. The approach we use for designing the algorithm for the case when the input graph is $2$-connected follows the approach of Heggernes et al.~\cite{treePathContract} for designing an \FPT\ algorithm for contracting to trees. Also, whenever we are dealing with an instance of \treelc\ we assume that we have an algorithm running in time
$\mathcal{O}((2\sqrt{\ell'} + 2)^{\mathcal{O}(k + \ell')} \cdot n^{\mathcal{O}(1)})$  for \treelpc, for every $\ell' <\ell$. That is, we give family of algorithms inductively for each $\ell' \in \mathbb{N}$, where the algorithm for \textsc{Tree Contraction} by Heggernes et al. forms the base case of our inductive hypothesis. 

We start with few observation regarding the graph class $\bbT$, which will be useful while designing the algorithm.

\begin{observation}
    \label{graph-class-prop}
    For each $T \in \bbT$ the following statements hold.
    \begin{enumerate}
        \item The chromatic number of $T$ is at most $\ncol$.
        \item If $T'$ is a graph obtained by subdividing an edge in $T$ then $T' \in \bbT$.
        \item If $T'$ is a graph obtained by contracting an edge in $T$ then $T' \in \bbT$.
    \end{enumerate}
\end{observation}
\begin{proof}
\emph{(Proof of Part 1.)} We first prove that for any graph $G$ with at least one edge, its chromatic number is upper bounded by $2\sqrt{|E(G)|}$. Let $C_1, C_2, \dots, C_q$ be the color classes in a proper coloring of $G$ which uses the minimum number of colors. Observe that there is at least one edge between $C_i, C_j$, where $i,j \in [q]$, $i \neq j$. This implies that $\binom{q}{2} \le |E(G)|$, which proves the claim. Next, consider $T_{\ell} \in \mathbb{T}_\ell$, and fix a spanning tree $T$ of $T_\ell$. Let $T' = E(T_{\ell})\setminus E(T)$. If $\ell > 0$ then from the claim above, we can properly color graph $T_{\ell}[V(T')]$ using at most $2\sqrt{\ell}$ many colors. Since $T_{\ell} - T'$ is a tree, we can properly color $T_{\ell}$ by coloring the vertices in $T_\ell -V(T')$ using two new colors.

\emph{(Proof of Part 2.)} For any connected graph $T$ if $|E(T)| \le |V(T)| - 1 + \ell$ then $T$ is contained in $\bbT$. Subdividing an edge adds a new vertex and an edge and hence this inequality is satisfied while maintaining the connectivity of graph. This implies $T'\in \bbT$, where $T'$ is obtained from $T$ by sub-dividing an edge in $T$.

\emph{(Proof of Part 3.)} Similar to the proof of part 2, contracting an edge decreases the number of vertices by one and number of edges by at least one. This implies $|E(T')| \le |V(T')| - 1 + \ell$. Contracting an edge maintains the connectivity of the graph and hence $T'\in \bbT$. 
  
\end{proof}

\begin{observation}\label{graph-class-prop2}
For a graph $T \in \bbT$, the graph $T' \in \bbT$ whenever $T'$ is obtained from $T$ as follows. Consider a vertex $v \in V(T)$, and a partition $N_1,N_2$ of $N_{T}(v)$. Let $V(T')= (V(T) \setminus \{v\}) \cup \{v_1,v_2\}$ and $E(T')= E(T-\{v\}) \cup \{(v_1, u) \mid u \in N_1\} \cup \{(v_2, u) \mid u \in N_2\} \cup \{(v_1,v_2)\}$. 
\end{observation}
\begin{proof}
Consider a vertex $v \in V(T)$, and a partition $N_1, N_2$ of $N_{T}(v)$. Let $V(T')= (V(T) \setminus \{v\}) \cup \{v_1,v_2\}$ and $E(T')= E(T-\{v\}) \cup \{(v_1, u) \mid u \in N_1\} \cup \{(v_2, u) \mid u \in N_2\} \cup \{(v_1,v_2)\}$. Notice that $T'$ is a connected graph. We have $|V(T')|= |V(T)|+1$ and $|E(T')|= |E(T)|+1\leq |V(T)|-1+\ell+1= |V(T')|-1 +\ell$. This concludes the proof. 
 
 \end{proof}


Let $(G,k)$ be an instance of \treelc. The measure we use for analysing the running time of our algorithm is $\mu= \mu(G,k)=k$. We start by applying some simple reduction rules. 
%
\begin{reduction rule} \label{rr:k-no-tco}
If $k < 0$ then return that $(G,k)$ is a {\rm \no} instance of {\rm \treelc}.
\end{reduction rule}

\begin{reduction rule} \label{rr:k-yes-rr-tco}
If $k = 0$ and $G \in \mathbb{T}_\ell$ then return that $(G,k)$ is a {\rm \yes} instance of {\rm \treelc}.
\end{reduction rule}

\begin{reduction rule} \label{rr:disconn-tco}
If $G$ is a disconnected or $k=0$ and $G \notin \bbT$ then return that $(G,k)$ is a {\rm \no} instance.
\end{reduction rule}

We assume that the input graph is $2$-connected, and design an algorithm for input restricted to $2$-connected graphs. Later, we will show how we can remove this constraint. The key idea behind the algorithm is to use a coloring of $V(G)$ with at most $\ncol$ colors to find a $T$-witness structure (if it exists) of $G$, where $G$ is contractible to $T \in \bbT$ using at most $k$ edge contractions (see Observation~\ref{graph-class-prop}). Moreover, if such a $T$ does not exist then we must correctly conclude that $(G,k)$ is a \no\ instance of \treelc. Towards this, we introduce the following notion.

 
\begin{definition}\label{def:compatibility-ws-tc}
Let $G$ be a $2$-connected graph, $T$ be a graph in $\bbT$, $\cal W$ be a $T$-witness structure of $G$, and $\phi: V(G) \rightarrow [\ncol]$ be a coloring of $V(G)$. Furthermore, let $T_S$ be a (fixed) spanning tree of $T$, $M=\{t,t' \mid tt' \in E(T) \setminus E(T_S)\} \cup \{t \in V(T) \mid d_{T}(t) \geq 3\}$, and $B=\{t \in V(T) \mid |W(t)| \geq 2\}$. We say that $\phi$ is $\cal W$-compatible if the following conditions are satisfied. 

\begin{enumerate}
\item For all $W \in \mathcal{W}$, and $w,w' \in W$ we have $\phi(w)=\phi(w')$. 
\item For all $t,t' \in M \cup B$ such that $tt' \in E(T)$ we have $\phi(W(t)) \neq \phi(W(t'))$. 
\item For all $t,t' \in M \cup B$ (not necessarily distinct), and a path $P=(t,t_1, \ldots, t_z,t')$, where $z \in \mathbb{N}$ such that for all $i \in [z]$ we have $t_i \notin M\cup B$ then $\phi(W(t)) \neq  \phi(W(t_1))$ and $\phi(W(t_z)) \neq \phi(W(t'))$.
\end{enumerate}
We refer to the set $M \cup B$ as the set of marked vertices. 
\end{definition}

Assume that $(G,k)$ is a \yes\ instance of \treelc, and $F$ be one of its (inclusion-wise) minimal solution. Furthermore, let $T=G/F$, and $\cal W$ be the $T$-witness structure of $G$. Suppose we are given $G$ and a $\cal W$ compatible coloring $\phi: V(G) \rightarrow [\ncol]$ of $G$, but we are neither given $\cal W$ nor $T$. We will show how we can compute a $T'$ witness structure ${\cal W}'$ of $G$ such that $|V(T')| \geq |V(T)|$, where $T' \in \bbT$. Informally, we will find such a witness structure by either concluding that none of the edges are part of the solution, some specific set of edges are part of the solution, or finding a star-like structure of the monochromatic components of size at least $2$ in $G$, with respect to $\phi$. Towards this, we will employ the algorithm for \cvcfull\ (\cvc) by Cygan~\cite{cvc-algo-cygan}.

\begin{proposition}[\cite{cvc-algo-cygan}]\label{lem:cvc-algo-cygan}\label{prop:cvc-algo}
\cvc\ admits an algorithm running in time $2^{k}n^{\OO(1)}$. Here, $k$ is the size of a solution and $n$ is the number of vertices in the input graph. 
\end{proposition}   

We note that we use the algorithm of Cygan~\cite{cvc-algo-cygan} instead of the algorithm by Cygan et al.~\cite{DBLP:conf/focs/CyganNPPRW11}, because the latter algorithm is a randomized algorithm. Also, the algorithm given by Proposition~\ref{prop:cvc-algo} can be used to output a solution. 

Consider the case when $G$ is $k$-contractable to a graph, say $T \in \bbT$, and let $\cal W$ be a $T$-witness structure of $G$. Furthermore, let $\phi: V(G) \rightarrow [\ncol]$ be a $\cal W$-compatible coloring of $G$, and $\cal X$ be the set of monochromatic components of $\phi$. We prove some lemmata showing useful properties of $\cal X$.

\begin{lemma}\label{lem:ws-calX-at-most-vertices}
Let $T'$ be the graph with $\cal X$ as the $T'$-witness structure of $G$. Then $T' \in \bbT$ and $|V(T')| \leq |V(T)|$.
\end{lemma}
\begin{proof}
  Every witness set of $\cal W$ is monochromatic with respect to $\phi$ (see item 1 of Definition~\ref{def:compatibility-ws-tc}).
  Therefore, for every $W \in {\cal W}$ there exists $X \in {\cal X}$ such that $W \subseteq X$. Moreover, by the definition of $\cal X$ we have that $G[X]$ is connected.
  Since $\bbT$ is closed under edge contraction, (item 3 of Observation~\ref{graph-class-prop}) therefore, $T' \in \bbT$ with $|V(T')| \leq |V(T)|$.
   
 \end{proof}

Next, we proceed to show how we can partition each $X \in {\cal X}$ into many smaller witness sets such that either we obtain $\cal W$ or a $T'$-witness structure of $G$ for some $T' \in \bbT$ which has at least as many vertices as $T$. Towards this, we introduce the following notions. 

For $X \in {\cal X}$, by $\hat X$ we denote the set of vertices that have a neighbor outside of $X$, \emph{i.e.} $\hat X= N(V(G) \setminus X)$. A \emph{shatter} of $X$ is a partition of $X$ into sets such that one of them is a connected vertex cover $C$ of $G[X]$ containing all the vertices in $\hat X$ and all other sets are of size $1$. The size of a shatter of $X$ is the of size of $C$. Furthermore, a shatter of $X$ is minimum if there is no other shatter with strictly smaller size. 

From Lemma~\ref{lem:ws-calX-at-most-vertices} (and Definition~\ref{def:compatibility-ws-tc}) it follows that for each $X \in {\cal X}$ there is ${\cal W}_X \subseteq {\cal W}$ such that $X = \cup_{Y \in {\cal W}_X} Y$. In the following lemma, we prove some properties of sets in ${\cal W}_X$, which will be useful in the algorithm design. 

\begin{lemma} \label{lem:contract-path}
Consider $X \in {\cal X}$ with $|X| \geq 2$, ${\cal W}_X \subseteq {}\cal W$ such that $X=\cup_{Y \in {\cal W}_X} Y$, and all of the following conditions are satisfied.
\begin{itemize}
\item $G[X]=(u, v_1, \ldots, v_q,v)$ is an induced path, where $q \in \mathbb{N}$.
\item For each $i \in [q]$ we have $deg(v_i)=2$. 
\item There exists $X' \in {\cal X}\setminus \{X\}$ such that $N(u) \cap X' \neq \emptyset$ and $N(v) \cap X' \neq \emptyset$. 
\end{itemize}
Then $|{\cal W}_X|=1$. 
\end{lemma} 
\begin{proof}
  Let $X=(u,v_1, v_2, \ldots, v_q, v)$, where for each $i \in [q]$ we have $deg_G(v_i)=2$. Also, let $X' \in {\cal X} \setminus \{X\}$ such that $N(u) \cap X' \neq \emptyset$ and $N(v) \cap X' \neq \emptyset$. Assume that $|{\cal W}_X| \geq 2$.
  Let $Y_1$ and $Y_2$ be the witness sets containing $u$ and $v$, respectively. Since, $|{\cal W}_X| \geq 2$, and each of the witness sets are connected therefore, we have $Y_1 \neq Y_2$. Notice that in $T$, for which $\cal W$ is a $T$-witness structure of $G$ there is a cycle $C$ containing $t_{Y_1}$, $t_{Y_2}$, and vertices corresponding to some of the witness sets in $X' \cup (X \setminus (Y_1 \cup Y_2))$. Here, $t_{Y_1}$ and $t_{Y_2}$ are vertices in $T$ such that $W(t_{Y_1})=Y_1$ and $W(t_{Y_2})=Y_2$. Notice that $C$ must contain at least $2$ marked vertices (see Definition~\ref{def:compatibility-ws-tc}). By definition, $X$ and $X'$ are monochromatic, and therefore, these marked vertices can not belong to $X$ or $X'$.
  Without loss of generality assume that $t_{Y_1}$ is one of the marked vertex and one of its neighbor, say $t'$, such that $W(t') \subseteq X'$ is another marked vertex on this cycle. This implies that $t_{Y_2}$ is contained in a path between two marked vertices namely, $t_{Y_1}$ and $t'$. But all the nodes on the path between $t_{Y_1}$ and $t_{Y_2}$ have the same color. This contradicts the fact that $\phi$ is a $\cal W$-compatible coloring of $G$ (see item 3 of Definition~\ref{def:compatibility-ws-tc}). 
 \end{proof}


\begin{lemma} \label{lem:excludes-path}
Consider $X \in {\cal X}$ with $|X| \geq 2$, ${\cal W}_X$ such that $X=\cup_{Y \in {\cal W}_X} Y$, and all the following conditions are satisfied.
\begin{itemize}
\item $G[X]=(v_0, v_1, \ldots, v_q, v)$ is an induced path, where $q \in \mathbb{N}$.
\item For each $i \in [q]$ we have $deg(v_i)=2$. 
\item There exists no $X' \in {\cal X}$ such that $N(u) \cap X' \neq \emptyset$ and $N(v) \cap X' \neq \emptyset$. 
\end{itemize}
Then $|{\cal W}_X|=|X|$. 
\end{lemma} 
\begin{proof}
  Recall that $F$ is a minimal solution corresponding to the witness structure $\cal W$. Assume that $|{\cal W}_X| < |X|$. 
  This implies that there exists $Y \in {\cal W}_X$ such that $|Y| \geq 2$. Let $t_Y$ be a vertex in $T$ such that $Y = W(t_Y)$. Also, let $v_i$ be the smallest $i$ such that $v_i \in Y$. Since $|Y| \geq 2$, $v_{i+1}$ is also present in $Y$. We can partition neighbors of $t_Y$ into $N_1$ and $N_2$ such that $N_1$ is adjacent to $v_i$ and $N_2$ is adjacent to $Y \setminus \{v_i\}$. By Observation~\ref{graph-class-prop2}, $G/(F \setminus \{(v_i, v_{i+1})\})$ is also a graph in $\bbT$. This contradicts the minimality of $F$. 
 
 \end{proof}


Next, we show that each $X \in {\cal X}$ for which Lemma~\ref{lem:contract-path} and~\ref{lem:excludes-path} are not applicable must contain exactly one big witness set. Moreover, the unique big witness set (together with other vertices as singleton sets) forms one of its shatters. 

\begin{lemma}\label{lem:unique-big-witness}
For $X \in {\cal X}$ with $|X| \geq 2$, let ${\cal W}_X \subseteq {\cal W}$ such that $X=\cup_{Y \in {\cal W}_X} Y$. Furthermore, the set $X$ does not satisfy the conditions of Lemma~\ref{lem:contract-path} or~\ref{lem:excludes-path}. Then there is exactly one big witness set in ${\cal W}_X$.
\end{lemma}
\begin{proof}
Consider $X \in {\cal X}$ with $|X| \geq 2$. Assuming a contradiction, suppose ${\cal W}_X$ contains two big witness sets say $Y$ and $Y'$. Notice that there cannot be an edge between a vertex in $Y$ and a vertex in $Y'$ (see item 2 of Definition~\ref{def:compatibility-ws-tc}). This together with the connectedness of $X$ implies that there is a path from a vertex in $y \in Y$ and a vertex in $y' \in Y'$ which contains a neighbor of $y$ in some $Z \in {\cal W}_X \setminus \{Y,Y'\}$. But then from item 2 and 3 of Definition~\ref{def:compatibility-ws-tc} we have $\phi(Y) \neq \phi(Z)$, a contradiction. Therefore, $X$ can contain at most one big witness set from $\cal W$. 

Suppose $X$ does not contain any big witness set.
Consider the case when $deg_G(v) = 2$ for all $v \in X$ and $X = V(G)$. Since all the witness sets in $X$ are singleton, there exists a cycle in $T$ such that all the vertices on this cycle have same color. This contradicts the fact that $\phi$ is $\cal W$-compatible (see item 3 Definition~\ref{def:compatibility-ws-tc}). 
We now consider case when $X$ is a proper subset of $V(G)$ or it contains a vertex of degree $3$.
Since $X$ does not satisfies conditions of Lemma~\ref{lem:contract-path} or~\ref{lem:excludes-path}, it is not an induced path or it is an induced path but one of its internal vertex has degree other than $2$.
Since $X$ is connected, in either case there exists $v \in X$ such that $deg_G(v) \ge 3$.
If $X$ contains all singleton witness set then $deg_T(t_v) \ge 3$ where $W(t_v) = \{v\}$.
Let $u \in X$ be a vertex adjacent to $v$ and $W(t_u) = \{u\}$. Since $|W(t_v)| = |W(t_u)| = 1$, neither $t_v$ nor $t_u$ is a cut vertex in $T$ which implies $deg_T(t_u) > 1$.
Let $t_v, t_1$ are two neighbors of $t_u$.
There exists a path between $t_v, t_1$ which does not contain vertex $t_u$. This implies there exists a cycle in $T$ containing $t_v, t_u, t_1$.
There are at least two vertices marked on this cycle.
Hence, either $t_u$ is marked or $t_u$ is contained between two marked vertices.
In either case, it contradicts the fact that $X$ is color class of a coloring which is $\cal W$-compatible (see item 2, 3 Definition~\ref{def:compatibility-ws-tc}).

\end{proof}

\begin{lemma}\label{lem:unique-shatter}
Consider $X \in {\cal X}$ such that $|X| \geq 2$ and it contains a big witness set, and it does not satisfy conditions of Lemma~\ref{lem:contract-path} or~\ref{lem:excludes-path}. Let ${\cal W}_X \subseteq {\cal W}$ such that $X=\cup_{Y \in {\cal W}_X} Y$, and $W^*$ be the (unique) big witness set in $X$. Then $W^*$ is a connected vertex cover of $G[X]$ and it contains $\hat{X}$. 
\end{lemma}
\begin{proof}
  Suppose $X$ contains a big witness set, say $W^*$. From Lemma~\ref{lem:unique-big-witness}, for each $Y \in {\cal W}_X \setminus \{W^*\}$ we have $|Y|=1$.
  We first prove that $W^*$ is a vertex cover of $G[X]$.
  Assume that $W^*$ is not a vertex cover of $G[X]$, then there is an edge $y_1y_2$ such that $y_1,y_2 \notin W^*$.
  For $i \in \{1, 2\}$, let $\{y_i\} = Y_i = W(t_i)$ .
  Since $G[X]$ is connected, there exists a path between $y_1, y_2$ and a vertex in $W^*$, which is contained in $X$.
  Without loss of generality, assume that there exist a path $P_1$ in $X$ from $y_1$ to $W^*$ which does not contain $y_2$.
  Since $G$ is $2$-connected, there exists a path, say $P_2$ from $y_2$ to a vertex in $W^*$, which does not contain $y_1$.
  Notice that there is a cycle in $T$ containing nodes $t^*, t_1, t_2$, where $W^* = W(t^*)$.
  At least two node of the vertices from this cycle must be marked, and  have different colors from each other.
  Hence, the path $P_2$ can not be contained in $X$. 
  We know that $t^*$ is contained in $M \cup B$. Let the other marked vertex in this cycle be $t$. The vertex $t$ is obtained by contracting some vertices on the path $P_2$. Notice that $t_1$ is vertex contained in the path between two vertices in $M \cup B$ and all the nodes in path from $t^*$ to $t_1$ has the same color. This contradicts the fact that $X$ is a color class in coloring which is $\cal W$-compatible (see item 2, 3 of Definition~\ref{def:compatibility-ws-tc}). Hence our assumption was wrong and no such edge $y_1y_2$ exits. Since, $W^*$ is a witness set, by definition, it is connected and therefore $W^*$ is a connected vertex cover of $G[X]$. Notice that all the above argument still holds if the path $P_1$ is simply an edge and $y_2$ is outside $X$. In other words, if there exists $y_1$ in $X \setminus W^*$ there is no edge $y_1 y_2$ such that $y_2$ is not contained in $X$. This implies that $\hat{X}$ is contained in $W^*$.
    
\end{proof}

Using Lemma~\ref{lem:excludes-path} to Lemma~\ref{lem:unique-shatter} we show how we can replace each $X \in {\cal X}$ with the sets of its shatter. Recall that we are given only $G$ and $\phi$, and therefore we know $\cal X$, but we do not know $\cal W$. In the Lemma~\ref{lem:replace-ws}, we show how we can find a $T'$-witness structure of $G$ for some $T' \in \bbT$, which has at least as many vertices as $T$ (without knowing $\cal W$).

\begin{lemma}\label{lem:replace-ws}
Given $\cal X$, we can obtain a $T'$-witness structure of $G$ in time $2^kn^{\OO(1)}$ time, where $T' \in \bbT$ and $|V(T')|\geq |V(T)|$. 
\end{lemma}

\begin{proof}
  Consider $X \in {\cal X}$. If $|X|=1$ then we let ${\cal W}_X=\{X\}$, which is the unique shatter of $X$.
  We now consider $X \in \cal X$ such that $|X| \ge 2$.
  If there is $X \in \cal X$ which satisfies the premise of Lemma~\ref{lem:contract-path} then contract all edges in $X$ and reduce $k$ by $|X| - 1$.
  If there exists $X \in \cal X$ which satisfies the premise of Lemma~\ref{lem:excludes-path} then replace $X$ in $\cal X$ with $|X|$ many singleton sets $\{v\}$ for each $v \in X$. 
  If there exists $X \in \cal X$ which does not satisfy the conditions of Lemma~\ref{lem:contract-path} and ~\ref{lem:excludes-path} then from Lemma~\ref{lem:unique-big-witness} we know that $X$ contains exactly one big-witness set, say $\hat W$.  
  Moreover, Lemma~\ref{lem:unique-shatter} implies that $\hat W$ is a connected vertex cover of $G[X]$ containing $\hat X$.
  In this case, we will find a shatter $W^*$ of $X$, which has size at most $|\hat W|$ as follows as follows. Let $G'$ be the graph obtained from $G[X]$ by adding a (new) vertex $v^*$ for each vertex $v \in \hat X$, and adding the edge $(v,v^*)$. Then we find a minimum sized connected vertex cover of $C$ of $G'$ by using the algorithm given by Proposition~\ref{lem:cvc-algo-cygan}.
  Notice that a minimum connected vertex cover of a graph does not contain any degree one vertex therefore, $\hat X \subseteq C$. From the definition of minimum shatter and the minimality of set $C$, it follows that ${\cal W}_X= \{C\} \cup \{\{x\} \mid x \in X \setminus C\}$ is a minimum shatter of $X$. Notice that apart from computing connected vertex cover, all other steps can be performed in polynomial time.
  Since the size of each witness set in $\cal W$ is bounded by $k+1$, therefore there exists a connected vertex cover of size at most $k + 1$.
  Moreover, we can compute connected vertex cover in time $2^{k+1}n^{\OO(1)}$ (Proposition~\ref{lem:cvc-algo-cygan}), and there are at most $n$ sets in $\cal X$. Therefore, the overall running time is bounded by $2^{k}n^{\OO(1)}$. 
    
\end{proof}

Now we are ready to present our randomized algorithm for \treelc\ when input graph is $2$-connected. 

\begin{theorem} \label{thm:2conn-algo}
  There is a Monte Carlo algorithm for solving
  \treelc\ on $2$-connected graphs running in time $\mathcal{O}((\ncol)^{\mathcal{O}(k + \ell)} \cdot n^{\mathcal{O}(1)})$, where $n$ is the number of vertices in the input graph. It does not return false positive and returns correct answer with probability at least $1 - 1/e$.
\end{theorem}

\begin{proof}
  Let $(G,k)$ be an instance of \treelc, where $G$ is a $2$-connected graph. Furthermore, the Reduction Rules~\ref{rr:k-no-tco} and~\ref{rr:disconn-tco} are not applicable, otherwise we can correctly decide whether or not $(G,k)$ is a \yes\ instance. The algorithm starts by computing a random coloring $\phi: V(G) \rightarrow [\ncol]$, by choosing a color for each vertex uniformly and independently at random. Let $\cal X$ be the set of monochromatic connected components with respect to $\phi$ in $G$.
  The algorithm applies Lemma~\ref{lem:replace-ws} in time $2^{k}n^{\OO(1)}$ and tries to compute $T'$ such that $T' \in \bbT$ and $G$ is $k$-contractible to $T'$.
  It runs $(\ncol)^{6k + 8\ell}$ many iterations of two steps mentioned above.
  If for any such iteration it obtains a desired $T'$-witness structure of $G$ then it returns \yes.
  If none of the iterations yield \yes\ then the algorithm returns \no. This completes the description of the algorithm.

Observe that the algorithm returns \yes\ only if it has found a $T' \in \bbT$ such that $G$ is contractible to $T'$ using at most $k$ edge contractions. Therefore, when it outputs \yes, then indeed $(G,k)$ is a \yes\ instance of \treelc. 
We now argue that if $(G, k)$ is a \yes\ instance then using a random coloring the algorithm (correctly) returns the answer with sufficiently high probability. 
Let $T$ be a graph in $\bbT$, such that $G$ is $k$-contractible to $T$, and $\cal W$ be a $T$-witness structure of $G$.
Furthermore, let $T_S$ be a (fixed) spanning tree of $T$, and vertex set $M$, $B$ are set of vertices defined in Definition~\ref{def:compatibility-ws-tc}.
Let $\psi : V(G) \rightarrow [\ncol]$ be a coloring where colors are chosen uniformly at random for each vertex.
The total number of vertices contained in big witness sets of $\calW$ is at most $2k$. 
By our assumption, every leaf is a singleton witness set and it is adjacent to a big witness set. Here, we assume that the number of vertices in $T$ is at least $3$, otherwise we can solve the problem in polynomial time. This implies that no leaf is in $M \cup B$.
Consider graph $T'$ obtained from $T$ by deleting all the leaves and deleting edges in $E(T_{\ell}) \setminus E(T_S)$.
All the marked vertices of $T_{\ell}$ and all the paths connecting two marked vertices are also present in $T'$.
Notice that $T'$ is tree with at most $k + 2\ell$ leaves.
Since the number of vertices of degree three is at most the number of leaves in any tree, there are at most $k + 2\ell$ vertices of degree at least $3$.
There are at most $k$ vertices in $T$ which are big witness sets and at most $2\ell$ vertices incident to edges in $E(T_{\ell}) \setminus E(T_S)$.
Hence the total number of marked vertices is at most $2k + 4\ell$.
Since $T'$ is a tree, there are at most $2k + 4\ell$ vertices which lie on a path between two vertices in $M \cup B$ and are adjacent to one of these.
The number of vertices of $G$ which are marked vertices or vertices which are adjacent to it in $T'$ is at most $2(2k + 4\ell) + 2k$.
Therefore, the probability that $\psi$ is compatible with $\calW$ is at least $1 / (\ncol)^{6k + 8\ell}$.
Since the algorithm runs $(\ncol)^{6k + 8\ell}$ many iterations, probability that none of these colorings which is generated uniformly at random is compatible with $\calW$ is at most
$(1 - 1/(\ncol)^{6k + 8\ell}) ^{(\ncol)^{6k + 8\ell}} < 1/e$.
Hence, algorithm returns a solution on positive instances with probability at least $1-1/e$.
Each iteration takes $2^k \cdot n^{\mathcal{O}(1)}$ time and hence the total running time of the algorithm is $\mathcal{O}((\ncol)^{\mathcal{O}(k + \ell)} \cdot n^{\mathcal{O}(1)})$.
 
 \end{proof}


Next, we design reduction rules and a branching rule whose (exhaustive) application will ensure that the instance of \treelc\ we are dealing with is $2$-connected. Either we apply one of these reduction rules or branching rule, or we resolve the instance using the algorithm for \treelpc, where $\ell'<\ell$. This together with Theorem~\ref{thm:2conn-algo} gives us an algorithm for \treelc\ on general graphs. 

\begin{lemma} \label{lemma:k-yes-tco}
If for some $0 \leq \ell'< \ell$, $(G,k)$ is a {\rm \yes} instance of \treelpc\ then return that $(G,k)$ is a {\rm \yes} instance of \treelc.
\end{lemma}

Our next reduction rule deals with vertices of degree of $1$.

\begin{reduction rule} \label{rr:deg-one-rule}
If there is $v\in V(G)$ such that $d(v)=1$ then delete $v$ from $G$. The resulting instance is $(G-\{v\}, k)$. 
\end{reduction rule}

If a connected graph $G$ is not $2$-connected graph then there is a cut vertex say, $v$ in $G$. Let $C_1, C_2, \ldots, C_t$ be the components of $G-\{v\}$. Furthermore, let $G_1= G[V(C_1) \cup \{v\}]$ and $G_2=G - V(C_1)$. Next, we try to resolve the instance (if possible) using the following lemma.

\begin{lemma} \label{rr:2-conn}
If there exists $\ell_1$ and $\ell_2$ with $\ell_1 + \ell_2 = \ell$, where $\ell_1, \ell_2 >0$, and $k_1$ and $k_2$ with $k_1 + k_2 = k$ such that $(G_1,k_1)$ is a {\rm \yes} instance of \treeloc\ and $(G_2,k_2)$ is a {\rm \yes} instance of \treeltc\ then return that $(G,k)$ is a {\rm\yes} instance of \treelc.
\end{lemma}

Notice that if Lemma~\ref{rr:2-conn} is not applicable then one of $G_1$ or $G_2$ must be contracted to a tree. Let $k_1$ be the smallest integer such that $(G_1,k_1)$ is a \yes\ instance of \treec, and $k_2$ be the smallest integer such that $(G_2,k_2)$ is a \yes\ instance of \treec. Notice that $k_1$ and $k_2$ can be computed in (deterministic) time $4^kn^{\OO(1)}$ using the algorithm for \treec\ ~\cite{treePathContract}. We next proceed with the following branching rule. 

\begin{branching rule} \label{br:2-conn}
We branch depending on which of the graphs among $G_1$ and $G_2$ are contracted to a tree. Therefore, we branch as follows.
\begin{itemize}
\item Contract $G_1$ to a tree, and the resulting instance is $(G_2, k-k_1)$.
\item Contract $G_2$ to a tree, and the resulting instance is $(G_1, k-k_2)$.
\end{itemize}
\end{branching rule}

Note that the measure strictly decreases in each of the branches of the Branching Rule~\ref{br:2-conn} since Reduction Rule~\ref{rr:deg-one-rule} is not applicable. If we are unable to resolve the instance using Lemma~\ref{lemma:k-yes-tco} and~\ref{rr:2-conn}, and Reduction Rules~\ref{rr:disconn-tco} and~\ref{rr:deg-one-rule} and Branching Rule~\ref{br:2-conn} are not applicable then the input graph is $2$-connected. And, then we resolve the instance using Theorem~\ref{thm:2conn-algo}.



\begin{theorem}
For each $\ell \in \mathbb{N}$, there is a Monte Carlo algorithm for solving \treelc\ with running in time $\mathcal{O}((\ncol)^{\mathcal{O}(k + \ell)} \cdot n^{\mathcal{O}(1)})$. It does not return false positive and returns correct answer with probability at least $1 - 1/e$.
\end{theorem}

\begin{proof}
Let $(G,k)$ be an instance of \treelc. If $G$ is $2$-connected then we resolve the instance using Theorem~\ref{thm:2conn-algo} with the desired probability bound. If $G$ is not connected then we correctly resolve the instance using Reduction Rule~\ref{rr:disconn-tco}. Moreover, the Reduction Rule~\ref{rr:disconn-tco} can be applied in polynomial time. Hereafter, we assume that $G$ is connected, but not $2$-connected. 

In this case, we proceed by either resolving the instance using Lemma~\ref{lemma:k-yes-tco} or Lemma~\ref{rr:2-conn}, or applying the Reduction Rule~\ref{rr:deg-one-rule}, or applying the Branching Rule~\ref{br:2-conn}. We prove the claim by induction on the measure $\mu=\mu(G,k)=k$. 

If $\ell=0$ then we can resolve the instance using the (deterministic) algorithm for \treelc\ in~\cite{treePathContract} in time $4^kn^{\OO(1)}$. We note here that though the deterministic algorithm presented in~\cite{treePathContract} has been mentioned to run in time $4.98^kn^{\OO(1)}$ but, it uses the algorithm for \textsc{Connected Vertex Cover} as a black-box, which has been improved in~\cite{cvc-algo-cygan}. This also improves the running time of the deterministic algorithm in~\cite{treePathContract}. Hereafter, we inductively assume that whenever we are dealing with an instance of \treelc, we have an algorithm for \treelpc\ with the desired runtime and success probability bound, where $0 \leq \ell' < \ell$. We note that this does not interfere with the probability computation since the only randomized step (recursively) in our algorithm is when we employ Theorem~\ref{thm:2conn-algo}, in which case we directly resolve the instance. 

 If $k \leq 0$ then we correctly resolve the instance using Reduction Rules~\ref{rr:k-no-tco} and~\ref{rr:k-yes-rr-tco}. If $(G,k)$ is a \yes\ instance of \treelpc, for some $0 \leq \ell' < \ell$ then we correctly conclude that $(G,k)$ is a \yes\ instance of \treelc. Moreover, we obtain the desired probability and runtime bound using the assumption of existence of an algorithm with desired properties for every $0 \leq \ell' <\ell$. If $k >0$, and there is a vertex of degree $1$ then we remove this vertex (in polynomial time) to obtain an equivalent instance using Reduction Rule~\ref{rr:deg-one-rule}. If none of the above are applicable the $G$ has a cut vertex say. We consider the following case. If Lemma~\ref{rr:2-conn} is applicable then we correctly resolve the instance in allowed running time with the desired success probability. This again relies on the existence of an algorithm for \treelpc\ with desired properties,  for every $0 \leq \ell' <\ell$. Otherwise, we know that Branching Rule~\ref{br:2-conn} must be applicable, where the measure drops at least by $1$ in each of the branches since Reduction Rule~\ref{rr:deg-one-rule} is not applicable. Moreover, when none of the Reduction Rules~\ref{rr:k-no-tco} to~\ref{rr:deg-one-rule} are applicable, we cannot resolve the instance using one of Lemma~\ref{lemma:k-yes-tco} and Lemma~\ref{rr:2-conn}, and Branching Rule~\ref{br:2-conn} is not applicable then the graph is $2$-connected, and we resolve the instance using Theorem~\ref{thm:2conn-algo}. Notice the number of nodes in the search tree is bounded by $2^{\OO(k)}$, all the reduction rules can be applied in polynomial time, and at the leaves of the search tree and at the internal nodes we require time which is bounded by $\mathcal{O}((\ncol)^{\mathcal{O}(k + \ell)} \cdot n^{\mathcal{O}(1)})$. Thus, we obtain the desired running time and probability bound. 
\end{proof}

\section{Derandomization}
In this section, we derandomize the algorithm presented in Section~\ref{sec:rand-fpt-algo-treelc}. Before proceeding forward we define the following important object of this section. 

\begin{definition}[Universal Family] A $(n, k, q)$-universal family is a collection $\calF$, of functions from $[n]$ to $[q]$ such that for each $S \subseteq [n]$ of size $k$ and a function $\phi : S \rightarrow [q]$, there exists function $f \in \calF$ such that $f|_S \equiv \phi$.
\end{definition}

Here, $f|_S$ denotes the function $f$ when restricted to the elements of $S$. For $q = 2$, the universal family defined above is called an $(n, k)$-universal set~\cite{Naor-splitters}. Hence, $(n, k, q)$-universal family is a generalization of $(n ,k)$-universal set. The main result of this section is the following theorem (Theorem~\ref{thm:uni-family}), which we use to derandomize the algorithm presented in Section~\ref{sec:rand-fpt-algo-treelc}.

\begin{theorem}\label{thm:uni-family} 
For any $n, k, q \ge 1$, one can construct an $(n, k, q)$-universal family of size $\mathcal{O}(q^k\cdot k^{\mathcal{O}(k)} \cdot \log n)$ in time $\mathcal{O}(q^k\cdot k^{\mathcal{O}(k)} \cdot n\log n)$.
\end{theorem}

Before proceeding to the proof of Theorem~\ref{thm:uni-family}, we state how we use it to derandomize the algorithm presented in Section~\ref{sec:rand-fpt-algo-treelc}. Let $(G,k)$ be an instance of \treelc. Assume that $(G,k)$ is a \yes\ instance of \treelc, and let $F$ be one of its solution. Furthermore, let $T=G/F$, where $T \in \bbT$ and $\cal W$ be the $T$-witness structure of $G$, and $\phi: V(G) \rightarrow [\ncol]$ be a $\cal W$-compatible coloring of $G$. Recall that our randomized algorithm starts by coloring vertices in $G$ uniformly and independently at random, and then uses this coloring to extract a witness structure out of each color classes. We then argued that any random coloring is ``equally good'' as that of $\phi$ with sufficiently high probability, which is given by a function of $k$ (and $\ell$). To derandomize this algorithm, we construct a family $\mathcal{F}$ of (coloring) functions from $[n]$ to $[\ncol]$. We argue that one of the colorings in the family that we compute is ``equally good'' as that of $\phi$. Recall that the number of vertices which we need to be colored in a specific way for a coloring to be $\cal W$-compatible is bounded by $6k+8\ell$ (see Definition~\ref{def:compatibility-ws-tc} and Theorem~\ref{thm:2conn-algo}). Let $S$ be the set of vertices in $G$ which needs to be colored in a specific way as per the requirements of Definition~\ref{def:compatibility-ws-tc}.
We can safely assume that $|S| = 6k + 8\ell$. If this is not the case we can add arbitrary vertices in $S$ to ensure this. Notice that any coloring $f$ of $G$ such that $f|_S=\phi|_{S}$ also satisfies the requirements of Definition~\ref{def:compatibility-ws-tc}. Let $\mathcal{F}$ be an $(n, 6k + 8\ell, \ncol)$-universal family constructed using Theorem~\ref{thm:uni-family}. Instead of using random coloring in the algorithm presented in Section~\ref{sec:rand-fpt-algo-treelc}, we can iterate over functions in $\calF$. Notice that we do not know $S$ but for any such $S$, we are guaranteed to find an appropriate coloring in one of the functions in $\calF$, which gives us the desired derandomization of the algorithm. 



In rest of the section, we focus on the prove of Theorem~\ref{thm:uni-family}. Overview of the proof is as follows: Let $S$ be a set of size $k$ in an $n$-sized universe $U$. We first \emph{reduce} this universe $U$ to another universe $U'$ whose size is bounded by $k^2$. We ensure that all elements of $S$ are mapped to different elements of $U'$ during this reduction. Let $Y$ be the range of $S$ in $U'$. We further partition $U'$ into $\log k$ parts such that $Y$ is almost equally divided among these partition. In other words, each partition contains (roughly) $k/\log k$ many elements of $Y$. For each of these parts, we explicitly store functions which represents all possible $q$-coloring of elements of $Y$ in this partition. Finally, we ``pull back'' these functions to obtain a coloring of $S$. 

\begin{definition}[Splitter~\cite{Naor-splitters}]\label{def:splitters} An $(n, k, q)$-splitter $\mathcal{F}$ is a family of functions from $[n]$ to $[q]$ such that for every set $S \subseteq [n]$ of size $k$ there exists a function $f \in \mathcal{F}$ that splits $S$ evenly. That is, for every $1 \le z, z' \le q$, $|f^{-1}(z) \cap S|$ and $|f^{-1}(z') \cap S|$ differ by at most 1.
\end{definition}


\begin{lemma}\label{lemma:splitter2} For every $1 \leq k,q \le n$ there is a family of $(n, k, q)$-splitter of size $\mathcal{O}(n^{\mathcal{O}(q)})$ which can be constructed in the same time.
\end{lemma}
\begin{proof} Let $x_0 = 0$ and $x_{q} = n$. For every choice of $q - 1$ elements in $[n]$ such that $1 \le x_1 < x_2 < \dots < x_{q-1} \le n$ define a function $f :[n] \rightarrow [q]$ as follows.
  For $x \in [n]$ we set $f(x) = j$ if $x_{j-1} < x \le x_{j}$ where $j \in [q]$.  This family has size $\binom{n}{q -1}$, and can be constructed in time $\mathcal{O}(n^{\mathcal{O}(q)})$.
\end{proof}

Following is another well known result for construction of splitter when $q = k^2$. We use this result to \emph{reduce} the size of the universe.

\begin{proposition}[\cite{saurabh-book,Naor-splitters}]\label{prop:splitter1} For any $n, k \ge 1$ one can construct an $(n, k, k^2)$-splitter of size $\mathcal{O}(k^{\mathcal{O}(1)} \log n)$ in time $\mathcal{O}(k^{\mathcal{O}(1)} n \log n)$.
\end{proposition}

Next, we look at the \textsc{$k$-Restriction} problem defined by Naor et al.~\cite{Naor-splitters}. Before defining the problem, we define some terminologies that will be useful. For a fixed set of alphabets, say $\{1, 2, \dots, b\}$ and a vector \emph{vector} $V$, which is an ordered collection of alphabets, the length of $V$ is the size of the collection. We represent $n$ length vector $V$ as $(v_1, v_2, \dots, v_n)$. For a positive integer $i \in [n]$, $V[i]$ denotes the alphabet at the $i^{th}$ position of $V$. Similarly, for an (index) set $S \subseteq [n]$, $V[S]$ denotes the $|S|$ sized vector obtained by taking alphabet at $i^{th}$ position in $V$, for each $i \in S$. In other words, if $S = \{i_1, i_2, \dots, i_k\}$ for $i_1 < i_2 < \dots < i_k$, then $V[S] = (V[i_1], V[i_2], \dots , V[i_k])$. An input to the \textsc{$k$-Restriction} problem is a set $\mathcal{C} = \{C_1, C_2, \dots, C_m\}$ called as a $k$-restrictions, where $C_j \subseteq [b]^k$ for $j \in [m]$ and an integer $n$. Here, $[b]^k$ denotes the set of all possible vectors of length $k$ over $[b]$, and $m$ denotes the size of the $k$-restrictions. We say that a collection $\mathcal{V}$ of vectors \emph{obeys} $\mathcal{C}$ if for all $S \subseteq [n]$ which is of size $k$ and for all $C_j \in \mathcal{C}$, there exists $V \in \mathcal{V}$ such that $V[S] \in C_j$. The goal of \textsc{$k$-Restriction} problem is to find a collection $\mathcal{V}$ of as small cardinality as possible, which obeys $\mathcal{C}$. Let $c = \min_{j \in [m]}|C_i|$, and let $T$ be the time needed to check whether or not the vector $V$ is in $C_j$. We next state the result of Naor et al.~\cite{Naor-splitters}, which will be useful for proving Theorem~\ref{thm:uni-family}.

\begin{proposition}[Theorem~1\ \cite{Naor-splitters}]\label{prop:k-restri} For any \textsc{$k$-Restriction} problem with $b \leq n$, there is a deterministic algorithm that outputs a collection obeying $k$-restrictions, which has size at most $(k \log n + \log m)/\log (b^k / (b^k - c))$. Moreover, the algorithm runs in time $\mathcal{O}\big( \frac{b^k}{c}\binom{n}{k} \cdot m\cdot T \cdot n^{k} \big)$. Here, $b$ is the size of the alphabet set, $m$ is the size of the $k$-restrictions, $n$ is the size of the vectors in the output set, and $c$ is the size of the smallest collection in the $k$-restrictions. 
\end{proposition}

Notice that a function from $[n]$ to $[q]$ can be seen as an $n$-length vector over the alphabet set $[q]$. Consider the case when each $C_j$ contains exactly one vector of length $k$ over $[q]$, \emph{i.e.} $\mathcal{C} = \{\{C\} \mid C \in [q]^k\}$, $m = q^k$, $c = 1$, and $T = \mathcal{O}(n)$. The output of \textsc{$k$-Restriction} on this input is exactly an $(n, k, q)$-universal family. Therefore, we obtain the following corollary. 

\begin{corollary} \label{cor:uni-family-brute} For any $n, k, q \ge 1$, one can construct an $(n, k, q)$-universal family of size $\mathcal{O}(q^k\cdot k \cdot (\log n + \log q))$ in time $\mathcal{O}(q^k \cdot n^{\mathcal{O}(k)})$.
\end{corollary}

Notice that we can not directly employ Corollary~\ref{cor:uni-family-brute} to construct the desired family, since its running time is $\mathcal{O}(q^k \cdot n^{\mathcal{O}(k)})$. Therefore, we carefully use splitter to construct an $(n, k, q)$-universal family to obtain the desired running time.

\emph{Proof of Theorem~\ref{thm:uni-family}.} 
  For the sake of clarity in the notations, we assume that $\log k$ and $k/\log k$ are integers.
  Let $\mathcal{A}$ be a $(n, k, k^2)$-splitter obtained by Proposition~\ref{prop:splitter1}.
  Let $\mathcal{B}$ be a $(k^2, k, \log k)$-splitter obtained by Lemma~\ref{lemma:splitter2}.
  Let $\mathcal{D}$ be a $(k^2, k/\log k, q)$-universal family obtained by Corollary~\ref{cor:uni-family-brute}. We construct $\mathcal{F}$ as follows. For every function $f_a$ in $\mathcal{A}$, $f_b$ in $\mathcal{B}$, and $\log k$ functions $g_1, g_2, \dots, g_{\log k}$ in $\cal D$, we construct a tuple $f=(f_a, f_b, g_1, g_2, \dots, g_{\log k})$, and add it to $\mathcal{F}$. We note here that $g_1, g_2, \dots, g_{\log k}$ need not be different functions.
  For $f \in \mathcal{F}$, we define $f : [n] \rightarrow [q]$ as follows. For $x \in [n]$, we have $f(x) = g_r(f_b(f_a(x)))$, where $r = f_b(f_a(x))$.

  We first argue about the size of $\cal F$ and the time needed to construct it. Notice that $|\mathcal{F}| \le |\mathcal{A}||\mathcal{B}||\mathcal{D}|^{\log k}$. We know $|\mathcal{A}| \le k^{\mathcal{O}(1)} \log n$, $|\mathcal{B}| \le \mathcal{O}(k^{\mathcal{O}(\log k)})$ and $|\mathcal{D}| \le q^{k / \log k} k^{\mathcal{O}(k / \log k)}$ by Proposition~\ref{prop:splitter1}, Lemma~\ref{lemma:splitter2}, and Corollary~\ref{cor:uni-family-brute}, respectively. This implies that $|\mathcal{F}| \in \mathcal{O}(q^{k} \cdot k^{\mathcal{O}(\log k)} \cdot \log n)$. Note that $\mathcal{A, B, D}$ can be constructed in time $\mathcal{O}(k^{\mathcal{O}(1)} n \log n)$, $\mathcal{O}(k^{\mathcal{O}(\log k)})$, and $\mathcal{O}(q^k \cdot k^{\mathcal{O}(k / \log k)})$, respectively. This implies that time required to construct $\mathcal{F}$ is bounded by $\mathcal{O}(q^{k} \cdot k^{\mathcal{O}(k)} \cdot n \log n)$.

  It remains to argue that $\mathcal{F}$ has the desired properties. Consider $S\subseteq [n]$ of size $k$ and $\phi: S \rightarrow [q]$. We prove that there exists a function $f \in \mathcal{F}$ such that $f|_S \equiv \phi$.
  By the definition of splitter, there exists $f_a \in \mathcal{A}$ such that $f_a$ evenly splits $S$ (see Definition~\ref{def:splitters}).
  Since $|S| < k^2$, for every $y \in [k^2]$, $|f_a^{-1}(y) \cap S|$ is either $0$ or $1$.
  Let $Y = \{y_1, y_2, \dots, y_{k}\}$ be a subset of $[k^2]$ such that $y_1 < y_2 < \dots < y_{k}$ and
  $|f_a^{-1}(y_i) \cap S| = 1$, for all $i \in [k]$.
  For $j = k/\log k$, we mark every $j^{th}$ element in set $Y$ marking $\log{k} - 1$ indices altogether. In other words, construct a subset $Y'$ of $Y$ of cardinality $\log{k} - 1$ such that $Y' = \{y_{1j}, y_{2j}, y_{3j} \dots, y_{(\log{k} - 1)j} \}$. We use the set $Y'$ to partition $[k^2]$ in a way that every partition contains almost $k/\log k$ many elements of $Y$. 
  Let $y_0 = 0$ and $y_{(\log k) j} = k^2$ and define set $Y_r = \{y \in Y \mid y_{r-1} < y \le y_{r}\}$ for $r \in [\log k]$. Recall that a $\mathcal{B}$ is $(k^2, k, \log k)$-splitter family obtained by Lemma~\ref{lemma:splitter2}.
  By construction, there exists a function $f_b$ which corresponds to subset $Y'$ of $\log{k} - 1$ many indices. In other words, there is a function $f_b$ such that $f_b^{-1}(r)$ contains all the elements in $Y_r$, for each $r$ in $[\log k]$.
  We note that size of $f_b^{-1}(r)$ could be as large as $k^2$. Recall that $\mathcal{D}$ is a $(k^2, k/\log k, q)$-universal family. Therefore, for every  $r \in [\log k]$ there exists $g_r \in \mathcal{D}$ such that $g_r|_{Y_r} \equiv \phi|_{Y_r}$.
  Consider a function $f=(f_a, f_b, g_1, g_2, \dots, g_{\log k})$ in $\mathcal{F}$ where $f_a, f_b$ and $g_r$ satisfies the property mentioned above. The function $f_a$ is bijective on $S$ and $f(S) = Y$. The function $f_b$ partitions $Y$ into $\log k$ many parts by mapping $Y$ into $Y_1, Y_2, \dots , Y_{\log k}$. For each $Y_r$ there exists a function $g_r$ which gives the desired coloring of elements in $Y_r$ and hence for the elements in $S$.
  Since we considering all possible combinations of $f_a, f_b$ and $\log k$ functions in $\mathcal{D}$, there exists a function $f$ such that $f|_S \equiv \phi$, which proves the theorem.  
\qed

\section{Non-existence of a Polynomial Kernel for \treelc} \label{sec:no-poly-kernel-treelc}
In this section, we show that \treelc\ does not admit a polynomial kernel unless \NP\ $\subseteq$ \CONPpoly. We note that \treec\
(\textsc{Tree Contraction}) does not admit a polynomial kernel unless \NP\ $\subseteq$ \CONPpoly~\cite{treePathContract}. We give a reduction from \treec\ to \treelc\ as follows.

\subparagraph*{Reduction.} Let $(G,k)$ be an instance of \treec. We create an instance $(G',k')$ of \treelc\ as follows. Initially, we have $G=G'$. Let $v^*$ be an arbitrarily chosen vertex in $V(G)$. For each $i \in [\ell]$, we add a cycle $(v^*,w^i_1, w^i_2, \ldots, w^i_{k+1})$ on $k+2$ vertices to $G'$, which pairwise intersect at $v^*$, and we set $k'=k$. This completes the description of the reduction. In the following lemma we establish equivalence between the two instances.

\begin{lemma}\label{lem:eq-treec-treelc}
$(G,k)$ is a \yes\ instance of \treec\ if and only if $(G',k')$ is a \yes\ instance of \treelc. 
\end{lemma}
\begin{proof}
In the forward direction, let $(G,k)$ be a \yes\ instance of \treec, and $S$ be one of its solution. Notice that $G'/S \in \mathbb{T}_\ell$, and $|S| \leq k'=k$. Therefore, $(G',k')$ is a \yes\ instance of \treelc. In the reverse direction, let $(G',k')$ be a \yes\ instance of \treelc, and $S$ be one of its (minimal) solution. Recall that for each $i \in [\ell]$ we have a cycle $C_i=(v^*,w^i_1, w^i_2, \ldots, w^i_{k+1})$ on $k+2$ vertices in $G'$, which pairwise intersect at $v^*$. This together with minimality of $|S|$ implies that $S \cap E(C_i)=\emptyset$. Furthermore, $G'[\{v\} \cup (\cup_{i \in [\ell]}V(C_i))]$ belongs to $\mathbb{T}_\ell \setminus \mathbb{T}_{\ell-1}$. Therefore, $G'[V(G)]/S$ must be a tree.
\end{proof}

\begin{theorem}\label{thm:no-poly-kernel-treelc}
\treelc\ does not admit a polynomial kernel unless {\rm \NP} $\subseteq$ {\rm \CONPpoly}.
\end{theorem}
\begin{proof}
Follows from construction of an instance $(G',k')$ of \treelc\ for a given instance $(G,k)$ of \treec, Lemma~\ref{lem:eq-treec-treelc}, existence of no polynomial kernel for \treec, and \NPC ness of \treelc. Here, we note that the reduction from \treec\ to \treelc\ is also a proof for \NPC ness of \treelc. 
\end{proof}

\section{PSAKS for \treelc}
In this section, we design a PSAKS for \treelc, which complements the result that \treelc\ does not admit a polynomial kernel assuming {\rm \NP} $\not \subseteq$ {\rm \CONPpoly} (Section~\ref{sec:no-poly-kernel-treelc}). 

Let $(G,k)$ be an instance of \treelc. The algorithm starts by applying Reduction Rules~\ref{rr:k-no-tco} to~\ref{rr:deg-one-rule} (if applicable, in that order). Next, we state the following lemma which will be useful in designing a reduction rule which will be employed for bounding the sizes of induced paths. 

\begin{lemma} \label{lemma:no-long-path} 
Let $(G,k)$ be an instance of \treelc\, and $P = (u_0, u_1, \ldots, u_q, u_{q+1})$ be a path in $G$, where $q\geq k+2$, and for each $i \in [q+1]$ we have $deg(u_i)=2$. Then no minimal solution $F$ to \treelc\ in $(G,k)$ with $|F| \leq k$ contains an edge incident to $V(P) \setminus \{u_0,u_{q+1}\}$.
\end{lemma}

\begin{proof}
Assume the contrary that $F$ contains at least one such edge. Observe that there are at least $k+1$ edges with endpoints in $V(P) \setminus \{u_0,u_{q+1}\}$. Therefore, there exists $i \in [q-1] \setminus \{1\}$ such that $u_{i-1}u_i \in F$ and $u_iu_{i+1} \notin F$, or $u_{i-1}u_i \notin F$ and $u_iu_{i+1} \in F$. Let us assume that there exists $i \in [q-1] \setminus \{1\}$ such that $u_{i-1}u_i \in F$ and $u_iu_{i+1} \notin F$ (other case is symmetric). Let $T=G/F$ with $V(T)=\{t_1,\cdots,t_p\}$, and $\mathcal{W}$ be the $T$-witness structure of $G$. Furthermore, let $t$ and $t'$ be the vertices in $T$ such that $u_{i-1}, u_{i} \in W(t)$ and $u_{i+1} \in W(t')$. If $t = t'$ then consider the following. Notice that $G[W(t)]$ is connected, $u_{i-1}, u_{i},u_{i+1} \in W(t)$, and $u_{i}u_{i+1} \notin F$. Therefore, $W(t)$ must contain the vertices of the sub-path $(u_{i+1},\ldots,u_q,u_{q+1})$ and the vertices of the subpath $(u_0,u_1,\ldots,u_{i-1},u_{i})$. But then, we have $|W(t)| > k+1$, a contradiction. Therefore, we have $t \neq t'$. Notice that $u_{i}$ is not a cut vertex in $G[W(t)]$, as there is exactly one edge incident on it. Therefore, $G[W(t) \setminus \{u_{i}\}]$ is connected. Let $\mathcal{W}' = (\mathcal{W} \setminus \{W(t)\} )\cup \{u_i\} \cup \{W(t) \setminus \{u_i\}\}$. Observe that $\mathcal{W}'$ is a partition of $V(G)$ which is a $G/F'$-witness structure of $G$, where $F'=F \setminus \{u_{i-1}u_i\}$. Here, $G/F'$ is the graph obtained by subdividing the edge $tt'$ in $T$, and by Observation~\ref{graph-class-prop}, $G/F'$ is also a graph in $\mathbb{T}_\ell$, which contradicts the minimality of $F$.
\end{proof}

Next, we design a reduction rule which will be useful in bounding length of induced paths whose internal vertices are of degree $2$.

\begin{reduction rule} \label{rr:no-long-path}
If $G$ has a path $P=(u_0,u_1,\dots,u_q,u_{q+1})$ such that $q > k+2$ and for all $i \in [q]$, we have $deg(u_i)=2$. Then contract the edge $u_{q-1} u_q$, \emph{i.e.} the resulting instance is $(G/\{u_{q-1} u_q\},k)$.
\end{reduction rule}

Note that Reduction Rule~\ref{rr:no-long-path} can be applied in polynomial time by searching for such a path (if it exists) in the subgraph induced on the vertices of degree $2$ in $G$. In the following lemma, we show that Reduction Rule~\ref{rr:no-long-path} is safe. 

\begin{lemma} 
Reduction Rule \ref{rr:no-long-path} is safe.
\end{lemma}

\begin{proof}
Let $P=(u_0,u_1,\dots,u_q,u_{q+1})$ be a path in $G$ such that $q > k+2$ and for all $i \in [q]$, we have $deg(u_i)=2$. Furthermore, let $G'=G/\{u_{q-1} u_q\}$, $P'=(u_0,u_1,\dots,u_{q-2},u^*,u_{q+1})$, where $u^*$ is the vertex resulting after contracting the edge $u_{q-1} u_q$. We consider the instances $(G,k)$ and $(G',k)$ of \treelc, and show that $\frac{\treelshort(G, k, F)}{\textsc{OPT}(G, k)} \le \frac{\treelshort(G', k', F')}{\textsc{OPT}(G', k')}$. Here, $\treelshort$ is a shorthand notation for the parameterized minimization problem for \treelc. 

Consider a minimal set $F' \subseteq E(G')$ such that $T'=G'/F'$ is in $\bbT$. If $|F'| \geq k + 1$, then the solution lifting algorithm returns $E(G)$, otherwise it returns $F = F'$. If $|F'| \geq k + 1$ then $\treelshort(G, k, F) \leq  k + 1 =\treelshort(G', k, F')$. Otherwise, let $V(T')=\{t_1,\cdots,t_r\}$ and $\mathcal{W}'$ denote the $T'$-witness structure of $G'$. By Lemma~\ref{lemma:no-long-path}, $F'$ has no edge incident on vertices in $V(P) \setminus \{u_0,u_{q+1}\}$. Therefore, every vertex in $V(P') \setminus \{u_0,u_{q+1}\}$ is in a singleton set of $\mathcal{W}'$. Let $\mathcal{W}=(\mathcal{W}' \setminus \{u^*\}) \cup \{\{u_{q-1}\},\{u_{q}\}\}$ to be a partition of $V(G)$. Then, $\mathcal{W}$ is a $T$-witness structure of $G$ where $T$ is $G/F$, which is obtained from $T'$ by subdividing an edges. From Observation~\ref{graph-class-prop}, $T$ is in $\bbT$. Therefore, $\treelshort(G, k, F) \le \treelshort(G', k, F')$.  

Next, consider an optimum solution $F^*$ to \treelc\ in $(G, k)$. If $|F^*| \geq k + 1$ then $\textsc{OPT}(G, k) = k + 1$ and by definition, $\textsc{OPT}(G', k) \le k + 1= \textsc{OPT}(G, k)$. Otherwise, we have $|F^*| \le k$. Let $T=G/F^*$, and $\mathcal{W}$ be the $T$-witness structure of $G$. By Lemma~\ref{lemma:no-long-path}, $F^*$ has no edge incident on $V(P) \setminus \{u_0,u_{q+1}\}$. Therefore, every vertex in $V(P) \setminus \{u_0,u_{q+1}\}$ is in a singleton set in $\mathcal{W}$. Let $\mathcal{W}'=({\cal W} \setminus \{\{u_{q-1}\},\{u_q\}\}) \cup \{\{u^*\}\}$ be a partition of $V(G')$. Then, $\mathcal{W}'$ is a $T'$-witness structure of $G'$, where $T'=G'/F^*$. Finally, $T'$ is the graph obtained from $T$ by contracting an edge. Hence, $T'\in \bbT$, and $\textsc{OPT}(G', k) \le \textsc{OPT}(G, k)$. Hence, we have $\frac{\treelshort(G, k, F)}{\textsc{OPT}(G, k)} \le \frac{\treelshort(G', k', F')}{\textsc{OPT}(G', k')}$.
\end{proof}

\begin{lemma} \label{lemma:treel-cvc}
  Consider an instance $(G, k)$ of \treelc\ on which Reduction Rule \ref{rr:no-long-path} is not applicable. If $(G, k)$ is a {\rm \yes} instance of {\rm \treelc} then $G$ has a connected vertex cover of size at most $2(k + 3)(k + 2\ell)$.
\end{lemma}
\begin{proof} 
Let $(G, k)$ be a {\rm \yes} instance of {\rm \treelc}, $F$ be one of its solution, $T=G/F$, where $T \in \bbT$, and $\mathcal{W}$ be the $T$-witness structure of $G$. Let $L$ be the set of leaves in $T$, and $X= V(T)\setminus L$. If $|V(T)|\leq 2$ then the claim trivially holds since $|F| \leq k$. Otherwise, we have $|V(T)| \geq 3$. In this case, by Lemma~\ref{lemma:treel-witness-prop1} we can assume that each vertex in $L$ belongs to a singleton witness set in $\cal W$. Notice that for $t_i,t_j \in L$, where $t_i \neq t_j$, and $W(t_i) = \{u\}$ and $W(t_j) = \{v\}$ we have $t_it_j \not\in E(T)$ (since $|V(T)| \geq 3$), and therefore $uv \not\in E(G)$. As $T[X]$ is connected, it follows that $S = \bigcup_{t \in X} W(t)$ is a connected vertex cover of $G$. We now argue that $|S|$ is at most $2(k + 3)(k + 2\ell)$.
  
Let $X_1\subseteq X$ be the set comprising of vertices in $T$ such that for each $t \in X_1$ we have $|W(t)| > 1$, and $X_2 = X \setminus X_1$.
  Since Reduction Rule~\ref{rr:deg-one-rule} is not applicable on $G$, we can assume that every leaf in $T$ is adjacent to a vertex in $X_1$. Notice that any connected induced subgraph of $T$ is in $\bbT$.
  Fix a spanning tree of $T - L$, and let $F$ be the set of edges which are not in this spanning tree. Since, $T - L \in \bbT$ therefore, we have $|F| \le \ell$. Next, we create a set of marked vertices $M$. We add both the endpoints of edges in $F$ to $M$, and add vertices in $X_1$ to $M$.
  Consider a graph $T'$ obtained from $T - L$ by deleting edges in $F$ and contracting all vertices with  degree exactly two in the graph $T - L$.
  It is easy to see that $T'$ is a tree with all its leaves marked and every internal vertex of degree at least 3.
  Hence the number of vertices in $T'$ is at most twice the number of marked vertices. Since there are at most $k + 2\ell$ marked vertices, we get $|V(T')| \le 2(k + 2\ell)$. Every edge in $E(T')$ corresponds to a simple path (or an edge) in $T$. Recall that the number of internal vertices in each such path is bounded by $k + 2$ as Reduction Rule \ref{rr:no-long-path} is not applicable. Hence, $|X_2|$ is at most $2(k + 2)(k + 2\ell)$. Since, there are at most $k$ more vertices in $W(t)$ for $t \in X_1$, $|S|$ is at most $2(k + 2)(k + 2\ell) + k$. This concludes the proof of lemma.
 
\end{proof}

Before describing the next reduction rule, we define a partition of $V(G)$ into the following sets.
$$H = \{u \in V(G) \mid deg(u) \ge 2(k + 3)(k + 2\ell) + 1 \}$$
$$I = \{v \in V(G) \setminus H \mid N(v) \subseteq H \}$$ 
$$R = V(G) \setminus (H \cup I) $$


Vertices $v, u$ are said to be \emph{false twins} if $N(v) = N(u)$. We use Lemma~\ref{lemma:treel-witness-prop2} to reduce the number of vertices in $I$ which have many false twins.
Let $G$ be $k$-contractible to a graph $T$ in $\bbT$ and $\mathcal{W}$ be the $T$-witness structure of $G$.

\begin{lemma}\label{lemma:treel-witness-prop2} 
Consider sets $X,U \subseteq V(G)$ such that $U$ is an independent set in $G$ and for all $v \in U$ we have $X \subseteq N(v)$. If $|U| \ge k + \ell + 2$ then there is a vertex $t \in V(T)$ such that $X \subseteq W(t)$. 
\end{lemma}
\begin{proof}
  We prove this by contradiction. Assume there exists $t \neq t'$ such that $X \cap W(t)$ and $X \cap W(t')$ are non-empty. Since $U$ is an independent set and $|U| \ge k + \ell + 2$, there are at least $\ell + 2$ vertices in $U$ which are not contained in any big witness sets.
  Consider the subgraph of $T$ (on at least $\ell+4$ vertices) induced on the vertices $\{t, t'\} \cup \{t_i \mid W(t_i) \text{ is a singelton witness set containing a vertex in } U \} $.
  After deleting any set of $\ell$ edges in $T$, there still exists a cycle in $T$. This is a contradiction the fact that $T \in \mathbb{T}_\ell$. 
\end{proof}

\begin{reduction rule}\label{rr:false-twins-treel} If there is a vertex $v \in I$ that has at least $k + \ell + 2$ false twins in $I$ then delete $v$, \emph{i.e.} the resulting instance is $(G - \{v\}, k)$.
\end{reduction rule}

\begin{lemma}
Reduction Rule \ref{rr:false-twins-treel} is safe.
\end{lemma}
\begin{proof} 
Let $v \in I$ such that $v$ has at least $k + \ell + 2$ false twins in $I$, and let $G'=G - \{v\}$. We consider instances $(G,k)$ and $(G',k)$ of \treelc, and show that $\frac{\treelshort(G, k, F)}{\textsc{OPT}(G, k)} \le \frac{\treelshort(G', k, F')}{\textsc{OPT}(G', k)}$. Here, $\treelshort$ is a shorthand notation for the parameterized minimization problem for \treelc. 

Consider a solution $F'$ to \treelc\ in $(G',k)$. If $|F'| \geq k + 1$ then the solution lifting algorithm returns $E(G)$, otherwise it returns $F = F'$. If $|F'| \geq k + 1$ then $\treelshort(G, k, F) \leq k+1 =\treelshort(G', k, F')$. Otherwise, $|F'| \le k$, and let $T'=G'/F$, where $T' \in \bbT$ with $\mathcal{W}'$ being the $T'$-witness structure of $G'$. Let $U$ be set of false twins of $v$ in $I$. Recall that $|U| \geq k + \ell + 2$. From Lemma~\ref{lemma:treel-witness-prop2}, there exists $t_i \in V(T')$ such that $N_{G'}(u_1) \subseteq W'(t_i)$ for $u_1$ in $U$. Let $T$ be the graph obtained from $T'$ by adding a new vertex $t_v$ as a leaf adjacent to $t_i$. Notice that $T\in \bbT$, which follows from the fact that $N_{G'}(u_1) = N_{G}(u_1) = N_G(v)$, and $N_G(u_1) \subseteq W'(t_i)$. Let $\mathcal{W} = \mathcal{W}' \cup \{\{v\}\}$ be a partition of $V(G)$. Then, $T$ is $G/F$ and $\mathcal{W}$ is the $T$-witness structure of $G$. Hence, $\treelshort(G, k, F) \le \treelshort(G', k, F')$.  

Next, consider an optimum solution $F^*$ to \treelc\ in $(G, k)$. If $|F^*| \geq k + 1$ then by definition, $\textsc{OPT}(G, k) \le k + 1= \textsc{OPT}(G, k)$. Otherwise, we have $|F^*| \le k$. Let $T=G/F^*$, and $\mathcal{W}^*$ denote the $T$-witness structure of $G$. By an argument analogous to the proof of $\treelshort(G, k, F) \le \treelshort(G', k', F')$, we know that there exists $t_j \in V(T)$ such that $N(v) \subseteq W(t_j)$. Let $t \in V(T)$ such that $v \in W(t)$. If $W(t) = \{v\}$ then $t$ is a leaf in $T$, which implies that $F^*$ is also a solution to \treelc\ in $(G',k)$, thus giving the desired relation. Otherwise, consider the following. Recall that $v$ has at least $k + \ell + 2$ false twins, and at least one of them, say $u$, belongs to a singleton witness set. That is, there exists a vertex $t'$ in $T$ such that $W(t') = \{u\}$. Let $\mathcal{W}'$ be the partition of $V(G)$ obtained from $\mathcal{W}^*$ by swapping the appearances of $u$ and $v$. Furthermore, let $F'$ be the set of edges obtained from $F$ by replacing each edge $xv$ with the edge $xu$, where for each $xv \in F$. Notice that $F'$ is also an optimal solution to \treelc\ in $(G, k)$, and a solution to \treelc\ in $(G', k)$. Therefore, $\textsc{OPT}(G', k) \le \textsc{OPT}(G, k)$. Hence, $\frac{\treelshort(G, k, F)}{\textsc{OPT}(G, k)} \le \frac{\treelshort(G', k, F')}{\textsc{OPT}(G', k)}$. 
\end{proof}

For $\alpha >1$, we let $d=\lceil {\frac{\alpha}{\alpha-1}}  \rceil$. Next, we state our last reduction rule. 

\begin{reduction rule} \label{rr:treel-large-nbd} If there are vertices $v_1, v_2, \cdots ,v_{k+\ell+2} \in I$ and $h_1, h_2, \cdots,$ $h_d \in H$ such that for all $i \in [k+\ell+2]$, we have $\{h_1,\dots, h_d \} \subseteq N(v_i)$ then contract all edges in $\tilde{E} = \{v_1h_i \mid i \in [d]\}$, and decrease $k$ by $d-1$. The resulting instance is $(G/\tilde{E}, k - d + 1)$.
\end{reduction rule}

We note that the \emph{lossy-ness} is introduced only in the Reduction Rule~\ref{rr:treel-large-nbd}. We have determined that $H' = \{h_1, h_2, \dots, h_d\}$ need to be in one witness bag but $G[H']$ may not be connected. To simplify the graph, we introduce additional vertex $v_1$ to the bag which contains $H'$. By doing this we are able to contract $H' \cup \{v_1\}$ into a single vertex. In the following lemma, we argue that the number of extra edge contracted in this process is $\alpha$ factor of the optimum solution. 

\begin{lemma} Reduction Rule~\ref{rr:treel-large-nbd} is $\alpha$-safe. 
\end{lemma}

\begin{proof}
Let $v_1, v_2, \cdots ,v_{k+\ell+2} \in I$ and $h_1, h_2, \cdots,$ $h_d \in H$ such that for all $i \in [k+\ell+2]$, we have $\{h_1,\dots, h_d \} \subseteq N(v_i)$. Furthermore, let $\tilde{E} = \{v_1h_i \mid i \in [d]\}$, $G'=G/{\tilde E}$, and $k'=k-d+1$. We consider instances $(G,k)$ and $(G',k')$ of \treelc, and show that $\frac{\treelshort(G, k, F)}{\textsc{OPT}(G, k)} \le \max \Big\{ \frac{\treelshort(G', k', F')}{\textsc{OPT}(G', k')}, \alpha \Big\}$. 

Consider a solution $F'$ of \treelc\ in $(G', k')$.  If $|F'| \geq k' + 1$, then the solution lifting algorithm returns $E(G)$, otherwise it returns $F = F' \cup \tilde{E}$. If $|F'| \geq k' + 1$ then $\treelshort(G', k', F') = k' + 1 = k - d$. In this case, $F = E(G)$ and  $\treelshort(G, k, F) \le k + 1=k'+d =\treelshort(G', k', F') + d - 1$. Next, consider the case when $|F'| \le k'$, and let $\mathcal{W}' = \{W'(t_1), W'(t_2), \dots , W'(t_q)\}$ be the $G'/F'$-witness structure of $G$. Let $w$ denote the vertex in $V(G')\setminus V(G)$ obtained by contracting the edges in $\tilde{E}$. Without loss of generality, assume that $w \in W'(t_1)$. Let $\mathcal{W}= (\mathcal{W}' \setminus \{W'(t_1)\}) \cup \{W_1\} $, where $W_1 = (W'(t_1)\setminus \{w\}) \cup \{v_1, h_1, h_2, \dots, h_d\} $. Note that $V(G) \setminus \{v_1, h_1, h_2, \dots, h_d\} = V(G') \setminus \{w\}$ and hence $\mathcal{W}$ is partition of $V(G)$. Furthermore, $G[W_1]$ is connected as $G'[W'(t_1)]$ is connected, and therefore, $E(G'[W_1 \setminus \{w\}]) \cup \tilde{E}$ contains a spanning tree of $G[W_1]$. Also, $|W_1| = |W'(t_1)| + d$, and any vertex which is adjacent to $w$ in $G'$ is adjacent to at least one vertex in $\{v_1, h_1, h_2, \dots, h_d\}$ in $G$. Thus, $\mathcal{W}'$ is a $G/F$-witness structure of $G$, where $G/F \in \bbT$. Therefore, $\treelshort(G, k, F) \le \treelshort(G', k', F') + d$.

Next, consider an optimum solution $F^*$ to \treelc\ in $(G, k)$, and let $T$ be $G/F^*$ with $\mathcal{W}$ being the $T$-witness structure of $G$. If $|F^*| \geq k + 1$, then $\textsc{OPT}(G, k) = k + 1= k'+d = \textsc{OPT}(G', k') + d -1$. Otherwise, we have $|F^*| \le k$, and there are at least $\ell + 3$ vertices, in $\{v_1, v_2, \dots , v_{k + \ell + 2}\}$ ($\subseteq I$) which are not in $V(F^*)$. That is, they are in singleton witness sets of $\mathcal{W}$. Then, by Lemma~\ref{lemma:treel-witness-prop2}, $\{h_1, h_2, \dots ,h_d\}$ are in the same witness set, say $W(t_i)$ where $t_i \in V(T)$. Consider the case when $v_1 \in W(t_i)$. Let $\tilde F$ be the edge set obtained from $F$ by replacing each edge $uv$ by $uw$, where $v \in \{v_1, h_1, \cdots, v_d\}$ and $u \notin \{v_1, h_1, \cdots, v_d\}$. Furthermore, let $F' = \tilde F \setminus \tilde{E}$. Notice that $|F'| \le |F^*| - d$, and $F'$ is solution to $(G', k')$. Therefore, $\textsc{OPT}(G', k') \leq |F^*| - d = \textsc{OPT}(G, k) - d$. Next, we consider the case when $v_1 \not\in W(t_i)$, and let $t_j \in V(T)$ be the vertex such that $v_1 \in W(t_j)$. Then, $t_i$ and $t_j$ are adjacent in $T$. Let $\mathcal{W}' = \mathcal{W} \cup \{W(t_{ij})\} \setminus \{W(t_i), W(t_j)\}$ of $V(G)$, where $W(t_{ij}) = W(t_i) \cup W(t_j)$. Clearly, $G[W(t_{ij})]$ is connected. Thus, $\mathcal{W}'$ is a $G/F$-witness structure of $G$, where $|F| = |F^*|+1$ as $|W(t_i)|-1+|W(t_j)|-1 = (|W(t_{ij})|-1)-1$. Furthermore, $F$ can be assumed to contain $\tilde{E}$, and therefore $F' = F \setminus \tilde{E}$ is solution to \treelc\ in $(G', k')$. This implies that $\textsc{OPT}(G', k') \le |F'| = |F^*|+1 - d = \textsc{OPT}(G, k) - d+1$. Thus, we have
\footnote{We use the bound, $\frac{x + p}{y + q} \le \max\{\frac{x}{y}, \frac{p}{q}\}$ for any positive real numbers $x, y, p, q$.}
$\frac{\treelshort(G, k, F)}{\textsc{OPT}(G, k)} \le \frac{\treelshort(G', k', F') + d}{\textsc{OPT}(G', k') + (d - 1)} \le \max \Big\{ \frac{\treelshort(G', k', F')}{\textsc{OPT}(G', k')}, \alpha \Big\}$. 
\end{proof}

\begin{lemma} \label{lemma:treel-approx-kernel} 
Let $(G, k)$ be an instance of \treelc\ where none of the Reduction Rules \ref{rr:no-long-path} to~\ref{rr:treel-large-nbd} are applicable. If $(G, k)$ is a \yes\ of \treelc\ then $|V(G)| \leq c[k(k + 2\ell)]^{d + 1}$, where $c$ is some fixed constant. 
\end{lemma}

\begin{proof} Since Reduction Rule~\ref{rr:no-long-path} is not applicable, from Lemma~\ref{lemma:treel-cvc} it follows that $G$ has a connected vertex cover $S$ of size at most $2(k + 3)(k + 2\ell)$. The set $H$ consists of vertices of degree at least $2(k + 3)(k + 2\ell) + 1$, and hence every vertex in $H$ is included in any connected vertex cover of $G$, which is of size at most $2(k + 3)(k + 2\ell)$. This implies that $|H| \le 2(k + 3)(k + 2\ell)$. Every vertex in $R$ has degree at most $2(k + 3)(k + 2\ell)$. Therefore, if $S \cap R$ is a vertex cover of $G[R]$, then $|E(G[R])|$ is bounded by $4(k + 3)^2(k + 2\ell)^2$. Also, by the definitions of $I$ and $R$, every vertex in $R$ has a neighbour in $R$. Therefore, there are no isolated vertices in $G[R]$. Thus, $|R|$ is bounded by $8(k + 3)^2(k + 2\ell)^2$. Now, we bound the size of $I$. For every set $H' \subseteq H$ of size at most $d$, there are at most $k + \ell + 2$ vertices in $I$ which have $H'$ as their neighbourhood. Otherwise, Reduction Rule~\ref{rr:false-twins-treel} would have been applicable. Hence, there are at most $(k + \ell + 2) \cdot \binom{2(k + 3)(k + 2\ell)}{d-1} $ vertices in $I$ which have degree at most $d$. A vertex in $I$ which is of degree at least $d+1$, is adjacent to all vertices in at least one subset of size $d$ of $H$. For a such a subset $H'$ of $H$, there are at most $k + \ell + 2$ vertices in $I$ which have $H'$ in their neighbourhood since Reduction Rule~\ref{rr:treel-large-nbd} is not applicable. Thus, there are at most $(k + \ell + 2) \binom{2(k + 3)(k + 2\ell)}{d}$ vertices in $I$ of degree at least $d$. Hence, $|I| \leq c' [k(k + 2\ell)]^{(d + 1)}$, for some fixed $c'$. Since $H \cup R$ is $\hat c k^2(k + 2\ell)^2$ (where $\hat c$ is a constant) and $d > 1$, the claim follows.
  
\end{proof}
 
\begin{theorem}\treelc\ admits a strict PSAKS, where the number of vertices is bounded by
$c[k(k + 2\ell)] ^{(\lceil {\frac{\alpha}{\alpha-1}\rceil + 1)}}$, where $c$ is some fixed constant. 
\end{theorem}

\begin{proof} Given $\alpha >1$, we choose $d=\lceil {\frac{\alpha}{\alpha-1}} \rceil$ and apply Reduction Rules \ref{rr:no-long-path}, \ref{rr:false-twins-treel} and \ref{rr:treel-large-nbd} on the instance as long as they are applicable. The reduction rules can be applied in $\mathcal{O}([k(k + 2\ell)]^{(d + 1)}n^{\OO(1)})$ time, where $n$ is the number of vertices in the input graph. If the number of vertices in resulting graph is more than $\mathcal{O}([k(k + 2\ell)]^{d + 1})$, then by Lemma~\ref{lemma:treel-approx-kernel} we have $\textsc{OPT}(G, k)=k + 1$ and the algorithm outputs $E(G)$ as a solution. Otherwise, it has $c[k(k + 2\ell)] ^{(\lceil {\frac{\alpha}{\alpha-1}\rceil + 1)}}$ vertices. 
\end{proof}



\bibliography{References-contract}


\end{document}